\RequirePackage{fix-cm}
\documentclass[smallextended]{svjour3}       
\smartqed  

\usepackage{amssymb,amsmath,amsfonts,mathrsfs}
\usepackage[utf8]{inputenc}
\usepackage[T1]{fontenc}
\usepackage{caption}
\usepackage[english]{babel}
\usepackage{graphicx}
\usepackage{color}
\usepackage{doi}
\usepackage{commath}
\usepackage{tablefootnote}
\usepackage{threeparttablex}
\usepackage{float}

\DeclareMathOperator{\RR}{\mathbb{R}}

\DeclareMathOperator{\QQ}{\mathbb{Q}}

\DeclareMathOperator{\ZZ}{\mathbb{Z}}

\DeclareMathOperator{\GC}{\mathcal{G}}

\DeclareMathOperator{\BC}{\mathcal{B}}

\DeclareMathOperator{\PC}{\mathcal{P}}
\DeclareMathOperator{\DC}{\mathcal{D}}
\DeclareMathOperator{\SC}{\mathcal{S}}

\DeclareMathOperator{\FC}{\mathcal{F}}

\DeclareMathOperator{\JC}{\mathcal{J}}
\DeclareMathOperator{\IC}{\mathcal{I}}

\DeclareMathOperator{\QC}{\mathcal{Q}}

\DeclareMathOperator{\GS}{\mathscr{G}}

\DeclareMathOperator{\PS}{\mathscr{P}}

\DeclareMathOperator{\QS}{\mathscr{Q}}

\DeclareMathOperator{\inth}{\Lambda}

\DeclareMathOperator{\size}{size}
\DeclareMathOperator{\poly}{poly}

\DeclareMathOperator{\dom}{dom}

\DeclareMathOperator{\const}{const}

\DeclareMathOperator{\diag}{diag}

\DeclareMathOperator{\BZero}{\mathbf 0}

\newcommand*{\intint}[2][1]{\{#1, \dots, #2\}}

\DeclareMathOperator{\SPO}{S\!P}
\DeclareMathOperator{\EVO}{E\!V}
\DeclareMathOperator{\DPR}{D\!P}

\title{Structured $(\min,+)$-Convolution And Its Applications For The Shortest Vector, Closest Vector, and Separable Nonlinear Knapsack Problems
\thanks{The article was prepared within the framework of the Basic Research Program at the National Research University Higher School of Economics (HSE).
}
}

\begin{document}

\author{D.~V.~Gribanov, I.~A.~Shumilov, D.~S.~Malyshev}


\institute{D.~V.~Gribanov \at National Research University Higher School of Economics, 25/12 Bolshaja Pecherskaja Ulitsa, Nizhny Novgorod, 603155, Russian Federation\\
\email{dimitry.gribanov@gmail.com}
\and
D.~S.~Malyshev \at National Research University Higher School of Economics, 25/12 Bolshaja Pecherskaja Ulitsa, Nizhny Novgorod, 603155, Russian Federation;
\at Huawei, Intelligent systems and Data science Technology center (2012 Laboratories), 7/9 Smolenskaya Square, Moscow, 121099, Russian Federation\\
\email{dsmalyshev@rambler.ru}
\and
I.~A.~Shumilov \at Lobachevsky State University of Nizhny Novgorod, 23 Gagarina Avenue, Nizhny Novgorod, 603950, Russian Federation\\
\email{ivan.a.shumilov@gmail.com}
}

\maketitle
\begin{abstract}
In this work we consider the problem of computing the $(\min, +)$-convolution of two sequences $a$ and $b$ of lengths $n$ and $m$, respectively, where $n \geq m$. We assume that $a$ is arbitrary, but $b_i = f(i)$, where $f(x) \colon [0,m) \to \RR$ is a function with one of the following properties:
\begin{enumerate}
    \item the linear case, when $f(x) =\beta + \alpha \cdot x$;
    \item the monotone case, when $f(i+1) \geq f(i)$, for any $i$;
    \item the convex case, when $f(i+1) - f(i) \geq f(i) - f(i-1)$, for any $i$;
    \item the concave case, when $f(i+1) - f(i) \leq f(i) - f(i-1)$, for any $i$;
    \item the piece-wise linear case, when $f(x)$ consist of $p$ linear pieces;
    \item the polynomial case, when $f \in \ZZ^d[x]$, for some fixed $d$.
\end{enumerate}

To the best of our knowledge, the cases 4--6 were not considered in literature before. We develop true sub-quadratic algorithms for them.

We apply our results to the knapsack problem with a separable nonlinear objective function, shortest lattice vector, and closest lattice vector problems.






\keywords{Convolution \and Nonlinear Knapsack \and Separable Objective \and Shortest Vector Problem \and Closest Vector Problem \and Dynamic Programming \and Integer Programming \and Piece-wise Linear}
\end{abstract}

\section{Introduction}\label{intro_sec}

\subsection{Structured $(\min,+)$-convolution}\label{conv_def_subseq}

The standard \emph{$(\min,+)$-convolution problem} is formulated in the following way. For given $a = \{a_0, a_1, \dots, a_{n-1}\}$ and $b = \{b_0, b_1, \dots, b_{m-1}\}$, where $n\geq m$, it is to compute $a \star b := c$, defined by the formula:
\begin{equation}\label{min_conv_def}\tag{MinConv}
    c_k = \min\limits_{i + j = k} \{a_i + b_j\}, \quad \text{for} \quad k \in \intint[0]{n+m-2}.
\end{equation}

In the current paper, it is more natural for us to work with another problem that is clearly linear-time equivalent to the original one. The problem is to compute $a \bullet b := c$, defined by the formula:
\begin{equation}\label{min_dot_def}\tag{ReducedMinConv}
    c_k = \min\limits_{0 \leq i \leq m-1} \{a_{k+i} + b_i\}, \quad \text{for} \quad k \in \intint[0]{n-m}.
\end{equation}
This formulation of \ref{min_conv_def} can be found, for example, in \cite{FasterAPSP_Circ}. We will call it \ref{min_dot_def}.

Unlike the standard $(+, \times )$-convolution, it is not known whether the  $(\min,+)$-convolution admits the existence of truly sub-quadratic algorithms. Moreover, the lack of truly sub-quadratic algorithms, despite considerable efforts, has led researchers to postulate the \emph{MinConv-hypothesis} that \ref{min_conv_def} cannot be solved in $O(n^{2 - \delta})$ time, for any constant $\delta > 0$ \cite{EquivConv,FineGrained_OneDim}. Many problems are known to have conditional lower bounds from the
MinConv hypothesis, see, e.g., \cite{TreeSparsity_NearLinear,SmallestRect_Rev,EquivConv,OnIPAndConv,FineGrained_OneDim,LRC14}. 

The trivial $O(n^2)$ running time can be improved to $n^2 / 2^{\Omega(\sqrt{\log n})}$ using a reduction to the $(\min, +)$-matrix product, due to Bremner et al. \cite{NecklacesConv}, and using the Williams' algorithm for all the pairs shortest path (APSP) problem \cite{FasterAPSP_Circ}, which was derandomized later by Chan and Williams \cite{DerandomRazbSmol}.

However, the $O(n^2)$-time barrier can be beaten for different special cases. Let $b_i = f(i)$, for some function $f \colon [0,m) \to \RR$. For some $f$, the computational complexity of \ref{min_conv_def} can be significantly reduced. In this paper, we consider the following cases:
\begin{enumerate}
    \item \emph{the linear case}, when $f(i) = \alpha \cdot i + \beta$, for fixed $\alpha,\beta \in \QQ$ and $i \in \intint[0]{m-1}$;
    \item \emph{the convex case}, when $f(i+1) - f(i) \geq f(i) - f(i-1)$, for $i \in \intint[1]{m-2}$;
    \item \emph{the concave case}, when $f(i+1) - f(i) \leq f(i) - f(i-1)$, for $i \in \intint[1]{m-2}$;
    \item \emph{the polynomial case}\footnote[1]{Actually, a more general class of rational functions $f(x)/g(x)$, where $f,g \in \ZZ^d[x]$, could be considered by the cost of using $2d$ instead of $d$ in the complexity bound. But for the sake of simplicity, this case is not considered.
    
    More generally, our approach is applicable to any functions $f(x)$ with a fixed number of inflection points.}, when $f(x) \in \ZZ^d[x]$, for some fixed $d$;
    \item \emph{the piece-wise linear case}, when $f(x)$ is represented by a piece-wise linear function with $p$ pieces;
    \item \emph{the monotone case}, when $a$ and $b$ are both monotone sequences and  $|a_i| = O(n), |b_i| = O(n)$.
\end{enumerate}

\begin{definition}\label{EVO_def}
For the second and third cases, we assume that $f$ is defined by \emph{the evaluation oracle}, which is denoted by $\EVO$ in our paper. Given $f$ and $x \in \dom(f)$, this oracle returns (if it is possible) the value of $f(x)$. For the sake of simplicity, we assume that $\EVO$ can also compare the values of given $f$ in pairs of different points $x,y \in \dom(f)$. 
\end{definition}

To the best of our knowledge, the third, fourth, and fifth cases are new, and we develop the first sub-quadratic algorithms for them in the current paper. For the third and fifth cases, we simultaneously estimate the oracle and arithmetic complexities. 

Despite that the linear case is a special variant of the convex case, we make a separate category for it, because of the beautiful folklore linear-time solution. This solution is based on the folklore queue data structure that additionally supports queries to a minimal element and contains a really low constant term, which is hiding inside the $O$-notation. Since we did not find a description of this solution in literature, we will give it in Subsection \ref{linear_conv_subs}.

Probably, an $O(n)$-time solution for the \ref{min_conv_def} linear case was firstly implicitly presented by Pferschy in \cite{KnapsackDPTrick}, where an $O(n \cdot W)$-time dynamic programming algorithm (DP-algorithm) for the bounded knapsack problem was presented. But, since the description of a knapsack DP-algorithm from \cite{KnapsackDPTrick} is quite complex, it is hard to extract an algorithm for \ref{min_conv_def} from it, and the folklore solution looks more natural and effective.

In the following Table \ref{conv_tb}, we emphasize the best known complexity bounds for the cases, mentioned above. Details could be found in Subsection \ref{plinear_subs} and in Theorems \ref{minconv_poly_th}, \ref{minconv_convex_th} of Subsection \ref{poly_subs}.
\begin{table}[ht!]
    \centering
    \begin{tabular}{||c|c|c||}
    \hline
    \hline
    Case: & Time: & Reference: \\
    \hline
    \hline
    General & $n^2 / 2^{\Omega(\sqrt{\log n})}$ &  Bremner et al. \cite{NecklacesConv} \&  Williams \cite{FasterAPSP_Circ} \\
    & & see also \cite{DerandomRazbSmol} \\
    \hline
    Monotone & $\tilde O(n^{1.5})$ & Shucheng et al. \cite{FasterMinConvMonotone} \\
    & & see also \cite{3SumViaAdditioveComb} \\
    \hline
    $f$ is linear  & $O(n)$ & folklore \\
     & & another variant due to Pferschy \cite{KnapsackDPTrick} \\
    \hline
    $f$ is convex   & $O(n)$ & Axiotis \& Tzamos \cite{GraphKnap} (2019) \\
    & $O\bigl(n \cdot \log(n)\bigr)$ & Kellerer and Pferschy \cite{ImprovedDP} (2004) \\
    \hline
    $f$ is concave   & $O(n^{4/3} \cdot \log^2(n))$ & {\color{red} this work} \\
    \hline
    $f$ is piece-wise linear & $O\bigl(p \cdot n \cdot \log(n)\bigr)$ & {\color{red} this work} \\
    & &  $p$ is a number of pieces \\
    \hline
    $f \in \ZZ^d[x]$ & $O\bigl(d^3 \cdot n^{1 + \frac{\sigma-1}{\sigma}} \cdot \log^2(n)\bigr)$ & {\color{red} this work} \\
    & & $\sigma = \log_2(d) + \frac{1}{1 + \log_2(d)}$ \\
    \hline
    Examples: & & \\
    $f \in \ZZ^2[x]$ & $\tilde O(n^{4/3})$ &\\
    $f \in \ZZ^3[x]$ & $\tilde O(n^{1.493})$ &\\
    $f \in \ZZ^4[x]$ & $\tilde O(n^{11/7})$ &\\
    \hline
    \hline
    \end{tabular}
    
    \captionsetup{justification=centering}
    \caption{Complexity bounds for \eqref{min_conv_def}}
    \label{conv_tb}
\end{table}

\begin{remark}\label{complexity_rm}
Everywhere in our paper, by the words "time", "running time" or "complexity" we mean the number of elementary arithmetic operations. Additionally, we guaranty that sizes of values, arising in intermediate computations, are respectively small and polynomially bounded with respect to the input size. 
Thus, in our work, we do not address the issue of numbers growth in intermediate calculations, unless such a thing is specifically highlighted.
As usual, we use the $\tilde O$-notation to hide logarithmic terms in complexity estimations.
\end{remark}

\subsection{ The nonlinear knapsack problem}\label{knap_def_subseq}

Let $w \in \ZZ_{>0}^{n}$, $W \in \ZZ_{>0}$, $u \in \ZZ^n_{>0}$, and $f \colon \ZZ^n \to \RR$ be \emph{a separable function}. That is $f(x) = f_1(x_1) + \dots + f_n(x_n)$, where $f_i(x) \colon \ZZ \to \RR$. Let us consider \emph{the bounded knapsack problem with a general separable objective function}, defined as follows:
\begin{gather}
f(x) \to \max \notag\\
\begin{cases}
w^\top x = W\\
0 \leq x \leq u\\
x \in \ZZ^n.
\end{cases}\label{knap_gen_def}\tag{KNAP-GEN}
\end{gather}

We are interested in the following special cases for $f(x)$ that define the corresponding problems:
\begin{gather}
    f(x) = c^\top x, \quad \text{for some} \quad c \in \ZZ_{>0}^n; \label{knap_lin_def}\tag{KNAP-LIN}\\
    f_i(x) \text{ are all convex functions}; \label{knap_mconv_def}\tag{KNAP-CONV}\\
    f_i(x) \text{ are all concave functions}; \label{knap_conv_def}\tag{KNAP-CONC}\\
    f_i(x) \text{ are all piece-wise linear, with $p$ pieces}; \label{knap_plin_def}\tag{KNAP-PLIN}\\
    f_i(x) \in \ZZ^d[x], \quad \text{for some fixed $d$}. \label{knap_poly_def}\tag{KNAP-POLY}
\end{gather}
The unbounded version of \ref{knap_lin_def}, when $u_i = +\infty$, for all $i$, will be denoted by U-KNAP-LIN. Note that any algorithm that solves a variant of the bounded knapsack problem also can be used to solve the unbounded one. Additionally, note that \ref{knap_poly_def} is a natural generalization of \ref{knap_lin_def}, since \ref{knap_poly_def} with $d=1$ is equivalent to \ref{knap_lin_def}. 

Put $w_{\max} = \|w\|_{\infty}$, $c_{\max} = \|c\|_{\infty}$, and $u_{\max} = \|u\|_{\infty}$.

\subsubsection{Some results about \ref{knap_lin_def}.}
The paper \cite{FastApproxKnap} gives a reduction of \ref{knap_lin_def} to $\{0,1\}$-knapsack of $O\bigl(n \cdot \log(u_{\max})\bigr)$ weights, bounded by $O(u_{\max} \cdot w_{\max})$. Together with the basic dynamic programming technique, due to Bellman \cite{Bellman}, it gives $\tilde O(n \cdot W)$-time algorithm for \ref{knap_lin_def}. The paper \cite{KnapsackDPTrick} removes the logarithmic term and gives an $O(n \cdot W)$-time algorithm. The linear-time algorithm for the monotone convex $(\min,+)$-convolution, due to \cite{GraphKnap}, together with the principle to put equivalent-weight items into buckets, due to \cite{ImprovedDP} (see also \cite{GraphKnap,KnapsackSubsetSum_SmallItems}), reduces the running time to $O(n + m \cdot W)$, where $m$ is the number of unique weights. Since $m \leq \min\{w_{\max}, n\}$, it gives an $O(n + w_{\max} \cdot W)$-time algorithm. 

The paper \cite{SteinitzILP} introduces an elegant proximity argument and uses it to give an $\tilde O(n \cdot w_{\max}^2)$-time algorithm. The work \cite{MultyKnapsack_Grib} combines the above proximity argument with the folklore linear $(\min,+)$-convolution algorithm and reduces logarithmic factors in the last complexity bound to give an $O(n \cdot w_{\max}^2)$-time algorithm. The same algorithm is presented in \cite{OnCanonicalProblems_Grib} for more general class of problems, which contains $\Delta$-modular simplicies and closed polyhedra. Finally, the paper \cite{KnapsackSubsetSum_SmallItems} carefully combines ideas of a part of the previous papers to give the state of the art $O(n + m \cdot w_{\max}^2)$-time algorithm for \ref{knap_lin_def}. The state of the art algorithms with different parametrizations by $c_{\max}$ are given in \cite{KnapPrediction,KnapsackSubsetSum_SmallItems}. The following Table \ref{lin_knap_tb} gives some comparison of the above results.

\begin{table}[h!]
    \centering
    \begin{tabular}{||c|c||}
    \hline
    \hline
    Time: & Reference: \\
    \hline
    $\tilde O(n \cdot W)$ & Bellman \cite{Bellman} \& Lawler \cite{FastApproxKnap} \\
    \hline
    $O(n \cdot W)$ & Pferschy \cite{KnapsackDPTrick} \\
    \hline
    $O(n + m \cdot W) = O(n + w_{\max} \cdot W)$ & Axiotis \& Tzamos \cite{GraphKnap}\\
    & see also \cite{ImprovedDP} and \cite{KnapsackSubsetSum_SmallItems} \\
    \hline
    $\tilde O(n \cdot w_{\max}^2)$ & Eisenbrand \& Weismantel \cite{SteinitzILP} \\
    \hline
    $O(n \cdot w_{\max}^2)$ & Gribanov \cite{MultyKnapsack_Grib} \\
    & see also \cite{OnCanonicalProblems_Grib}\\
    \hline
    $O(n + m \cdot w_{\max}^2) = O(n + w_{\max}^3)$ & Polak, Rohwedder \& W\k{e}grzycki \cite{KnapsackSubsetSum_SmallItems} \\
    & $m \leq \min\{n,w_{\max}\}$ is a number of unique weights \\
    \hline
    \hline
    \end{tabular}
    
    \captionsetup{justification=centering}
    \caption{Complexity bounds for \ref{knap_lin_def} with respect to $W$ and $w_{\max}$ }
    \label{lin_knap_tb}
\end{table}

\subsubsection{Nonlinear separable objective function.} Many tricks, developed for \ref{knap_lin_def}, do not work in this case. To the best of our knowledge, the best known algorithm, parameterized by $W$ or $w_{\max}$ that we can apply for the problem \ref{knap_gen_def}, is a straightforward application of the Bellman's DP-principle \cite{Bellman}  that gives an $O(n \cdot W^2)$-time algorithm. A straightforward application of the linear-time monotone convex $(\min,+)$-convolution algorithm, due to \cite{GraphKnap}, gives an $O(n \cdot W)$-time algorithm for \ref{knap_mconv_def}. 
We can use different variants of the $(\min,+)$-convolution, considered in our paper, to construct pseudopolynomial algorithms for the problems \ref{knap_conv_def}, \ref{knap_plin_def}, and \ref{knap_poly_def}. Related results are given in the following Table \ref{nonlin_kanp_tb} (details could be found in Theorem \ref{main_knap_comp_th}):

\begin{table}[h!]
    \centering
    \begin{tabular}{||c|c|c||}
    \hline
    \hline
    Case: & Time: & Reference: \\
    \hline
    \hline
    \ref{knap_gen_def} & $O(n \cdot W^2)$ & Bellman \cite{Bellman} \\
    \hline
    \ref{knap_mconv_def} & $O(n \cdot W)$ & Axiotis \& Tzamos \cite{GraphKnap} \\
    \hline
    \ref{knap_conv_def}  & $\tilde O(n \cdot W^{4/3})$ & {\color{red} this work} \\
    \hline
    \ref{knap_plin_def}   & $\tilde O\bigl(p \cdot n \cdot W \bigr)$ & {\color{red} this work} \\
    \hline
    \ref{knap_poly_def}   & $\tilde O(d^3 \cdot n \cdot W^{1 + \frac{\sigma-1}{\sigma}})$ & {\color{red} this work} \\
    & $\sigma = \log_2(d) + \frac{1}{1 + \log_2(d)}$ & \\
    \hline
    Examples: & & \\
    $d = 2$ & $\tilde O(n \cdot W^{4/3})$ &\\
    $d = 3$ & $\tilde O(n \cdot W^{1.493})$ &\\
    $d = 4$ & $\tilde O(n \cdot W^{11/7})$ &\\
    \hline
    \hline
    \end{tabular}
    
    \captionsetup{justification=centering}
    \caption{Complexity bounds for \ref{knap_lin_def}, the nonlinear cases}
    \label{nonlin_kanp_tb}
\end{table}

\begin{remark}
The obtained results can be extended for working with \emph{the multidimensional knapsack problem} without significant effort. Since the analysis is straightforward with respect to the $1$-dimensional case, the complexity bounds are the same as in the previous table with the only one difference: $W$ need to be replaced by $W^s$, where $s$ is the knapsack dimension parameter.
\end{remark}

\subsubsection{Some other related results on unbounded \ref{knap_lin_def}} The paper of Nesterov \cite{knapsack_nesterov} uses a beautiful method of generating functions to give an \mbox{$\tilde O(n \cdot w_{\max} + W)$}-time algorithm for U-KNAP-LIN. The paper, due to Jansen \& Rohwedder \cite{OnIPAndConv}, presents an $O(w_{\max}^2)$-time algorithm that can be additionally translated on general ILP problems. Despite the fact that the Nesterov's result is not very well known to the community, for $W = O(w^{2-\varepsilon}_{\max})$ and $n = O(w^{1-\varepsilon}_{\max})$ it outperforms the $O(w^2_{\max})$-complexity result. Additionally, note that, due to proximity (see \cite{SteinitzILP}) or periodicity reasons (see \cite{DiscConvILP}), it can be assumed that $W \leq w_{\max}^2$. In the natural assumption that $n \leq w_{\max}$, we have $\tilde O(n \cdot w_{\max} + W) = \tilde O(w^2_{\max})$. Thus, the Nesterov's complexity bound loses only a logarithmic factor with respect to the complexity bound, due to Jansen \& Rohwedder.

Another parameterizations could also be used. For example, the paper \cite{OnCanonicalProblems_Grib} gives an $O(n \cdot w_{opt}^2)$-time algorithm, where $w_{opt}$ is the weight of an item with the optimal relative cost:
$$
          \frac{c_j}{w_j} = \min_{i \in \intint n} \left\{ \frac{c_i}{w_i} \right\} \quad\text{and}\quad w_{opt} := w_j.
$$
Since $w_{opt} \leq w_{\max}$, this is a weaker parametrization than the parametrization by $w_{\max}$. Note that for $w_{opt} = o(w_{\max}/\sqrt{n})$, it outperforms the $O(w_{\max}^2)$-complexity bound \cite{OnIPAndConv}, due to Jansen \& Rohwedder.
Additionally, it is known, due to \cite{OnCanonicalProblems_Grib}, that if $W \geq w^2_{opt}$, then U-KNAP-LIN can be solved by a more efficient \mbox{$O\bigl(n \cdot w_{opt} \cdot \log(w_{opt})\bigr)$}-time algorithm.

The paper \cite{FasterKnap_MConv} uses the fast monotone $(\min,+)$-convolution, due to \cite{FasterMinConvMonotone}, and presents an $\tilde O\bigl(n + (c_{\max} + w_{\max})^{1.5} \bigr)$-time algorithm for U-KNAP-LIN, which outperforms an $O(w^2_{\max})$-time algorithm for $c_{\max} = O(w_{\max}^{4/3})$. The following Table \ref{uknap_tb} gives some comparison of the above results.

\begin{table}[h!]
    \centering
    \begin{tabular}{||c|c||}
    \hline
    \hline
    Time: & Reference: \\
    \hline
    $\tilde O(n \cdot w_{\max}^2 + W) = \tilde O(w^2_{\max})$ & Nesterov \cite{knapsack_nesterov} \\
    \hline
    $O(w^2_{\max})$ & Jansen \& Rohwedder \cite{OnIPAndConv} \\
    \hline
    $O(n \cdot w_{opt}^2)$ & Gribanov et al. \cite{OnCanonicalProblems_Grib}\\
    \hline
    $\tilde O\bigl(n + (c_{\max} + w_{\max})^{1.5} \bigr)$ & Bringmann \& Cassis \cite{FasterKnap_MConv} \\
    \hline
    \hline
    \end{tabular}
    
    \captionsetup{justification=centering}
    \caption{Complexity bounds for U-KNAP-LIN }
    \label{uknap_tb}
\end{table}

\subsection{The shortest and closest vector problems}\label{SVP_def_subseq}

Let $A \in \ZZ^n$, $\Delta = \abs{\det(A)} > 0$, $q \in \QQ^n$, and $\inth(A) = \{A t \colon t \in \ZZ^n\}$. Clearly, $\inth(A)$ is a full rank integer lattice with determinant $\Delta$. \emph{The shortest lattice vector problem} and the \emph{the closest lattice vector problem} with respect to the $l_p$-norm can be formulated as follows:
\begin{equation}\label{SVP_def}\tag{SVP}
    \min\bigl\{ \norm{x}_p \colon x \in \inth(A) \setminus \{\BZero\}\bigr\},
\end{equation}
\begin{equation}\label{CVP_def}\tag{CVP}
    \min\bigl\{ \norm{x - q}_p \colon x \in \inth(A)\bigr\}.
\end{equation}

In our paper, we consider only exact algorithms for \ref{SVP_def} and \ref{CVP_def} with theoretically provable complexity bounds. So, we avoid many works about approximate solutions or efficient practical algorithms. For a recent survey, see \cite{SVP_survey_2021}, see also \cite{SVP_survey_2011}. 

Exact algorithms for \ref{SVP_def} and \ref{CVP_def} have a rich history. The first direction of enumeration-based algorithms dates back to the papers of Pohst \cite{PohstSVP}, Kannan \cite{KannanSVP}, and Fincke \& Pohst \cite{FinckePohstSVP}. The Kannan's paper \cite{KannanSVP} gives an $n^{O(n)}$-time algorithm for \ref{SVP_def} and \ref{CVP_def}, and many others improved upon his technique to achieve better running times \cite{FinckePohstSVP,ImprovedKannan_Hanrot,HelfSVP,FastKannanSVP_Micc}. An important feature of these algorithms is that they are of polynomial space. To the best of our knowledge, the state of the art complexity bound $n^{\frac{n}{2 e} + o(n)}$ for \ref{SVP_def} via enumeration-based approach is given in \cite{FastKannanSVP_Micc}, due to Micciancio \& Walter.

Another direction is sieving-based algorithms for \ref{SVP_def}. It is dated to the seminal paper of Ajtai, Kumar \& Sivakumar \cite{AKS_2001}. The algorithms from this class have better theoretical running time $2^{O(n)}$ with respect to enumeration-base algorithms, but use exponential $2^{O(n)}$ space. Many extensions and improvements of sieving technique have been proposed in \cite{SVP_DiscreteGauss,AKS_2002,some_sieving_2008,SVP_Blomer_2009,SVP_survey_2011,SVP_Graphs,MV_2010,SVP_Nguyen2008,PS_ImprovedAKS_2009}. The paper \cite{SVP_DiscreteGauss}, due to Aggarwal, Dadush \& Regev, gives the state of the art $2^{n + o(n)}$-complexity bound. In fact, this paper solves the more general \emph{Discrete Gaussian Sampling} (DGS) problem. Note that above papers do not give single-exponential $2^{O(n)}$-time algorithms for \ref{CVP_def}. The first paper that generalizes the sieving approach to solve \ref{CVP_def} in single-exponential time is due to Aggarwal, Dadush \& Stephens-Davidowitz \cite{CVP_DiscreteGauss}. This paper extends the DGS sampling technique from \cite{SVP_DiscreteGauss} and solves \ref{CVP_def} by an $2^{n + o(n)}$-time algorithm.

The last direction that we want to refer concerns algorithms, which use \emph{the Voronoi cell of a lattice} -- the centrally symmetric polytope, corresponding to the points closer to the origin than to any other lattice point. This direction is started from the paper \cite{CVP_Slicing}, due to Sommer, Feder \& Shalvi. The seminal work of Micciancio \& Voulgaris \cite{SVP_exp}, which is built upon the approach of \cite{CVP_Slicing}, gives first known deterministic single exponential time algorithms for \ref{SVP_def} and \ref{CVP_def}. More precisely, it gives $2^{2n + o(n)}$-time algorithms. The space usage is $2^{n + o(n)}$. 

The existence of $2^{O(n)}$-time polynomial-space exact algorithms for \ref{SVP_def} or \ref{CVP_def} is the major open problem in the lattice algorithms field.

The works, mentioned above, are mainly concerned with the Euclidean norm $\norm{\cdot}_2$. Some results about \ref{SVP_def}-solvers for other norms are presented, for example, in \cite{SVP_Blomer_2009,DadushDis,DadushFDim,CoveringCubes}. The paper \cite{CoveringCubes} (see also the monograph \cite{DadushDis}) presents a general technique to extend any Euclidean norm solvers to arbitrary norms with an additional $2^{O(n)}$-time multiplicative factor.

\medskip
Now, let us discuss our motivation. All the algorithms, mentioned above, are \emph{fixed polynomial tractable} (FPT) with respect to the dimension $n$ parameter. In other words, a complexity bound of any of the algorithms above looks like $f(n) \cdot \poly(\size(A, q))$, where $f(n)$ is a computable function, depending only on $n$, and $\size(A,q)$ is the input encoding size. Is it possible to choose another parameterization? For example, could we build an algorithm, parameterized by the lattice determinant $\Delta$? The paper \cite{FPT_Grib} answers positively and gives an \mbox{$O\bigl(n^{\omega} \cdot \log(\Delta) + n \cdot \Delta^2 \cdot \log(\Delta) \bigr)$}-time dynamic programming algorithm for both \ref{SVP_def} and \ref{CVP_def} with respect to any $\norm{\cdot}_p$, for $p \geq 1$, where $w$ is the matrix multiplication exponent. In our work, we improve this running time to $O\bigl( n^{\omega}\cdot \log(\Delta) + m \cdot \Delta \cdot \log(\Delta) \bigr)$, where $m = \min\{n,\Delta\}$. The improvement consists just in careful using of the linear-time monotone convex $(\min,+)$-convolution algorithm, due to \cite{GraphKnap}. Our algorithm uses $O(\Delta)$ space.

Strictly speaking, we solve the following slightly more general problems than \ref{SVP_def} and \ref{CVP_def}.

Let $f \colon \ZZ_{\geq 0} \to \RR$ be a monotone and convex function. We define \emph{the generalized shortest lattice vector problem} in the following way:
\begin{equation}\tag{GENERALIZED-SVP}\label{gen_SVP_def}
    \min\Bigl\{\sum\limits_{i=1}^n f\bigl(\abs{x_i}\bigr) \colon x \in \inth(A)\setminus\{\BZero\} \Bigr\}.
\end{equation}
Clearly, the original $\ref{SVP_def}$ problem is equivalent to \ref{gen_SVP_def} with $f(x) = x^p$.

We also define \emph{the generalized closest vector problem} in the following way:
\begin{equation}\tag{GENERALIZED-CVP}\label{gen_CVP_def}
    \min\Bigl\{\sum\limits_{i=1}^n f\bigl(\abs{x_i - q_i}\bigr) \colon x \in \inth(A) \Bigr\}.
\end{equation}
Again, the original \ref{CVP_def} problem is equivalent to \ref{gen_CVP_def} with $f(x) = x^p$. 

In the following Table \ref{SVP_tb}, we group the state of the art results for different cases, mentioned above, together with our new result. All the algorithms from the table below are deterministic, except the sieving-based algorithm. Details could be found
in Theorems \ref{main_SVP_th} and \ref{main_CVP_th}.

\begin{table}[h!]
    \centering
    \begin{tabular}{||c|c|c|c||}
    \hline
    \hline
    Technique: & Time: & Space: & Reference: \\
    \hline
    \hline
    Enumeration & $n^{\frac{n}{2 e} + o(n)}$ & $n^{O(1)}$ & Micciancio \& Walter \cite{FastKannanSVP_Micc} \\
    \hline
    Sieving & $2^{n + o(n)}$ & $2^{n + o(n)}$ & Aggarwal, Dadush, Regev \& Stephens-Davidowitz \cite{SVP_DiscreteGauss,CVP_DiscreteGauss} \\
    \hline
    Voronoi cell & $2^{2n + o(n)}$ & $2^{n + o(n)}$ & Micciancio \& Voulgaris \cite{SVP_exp} \\
    \hline
    DP & $\tilde O(n^{\omega} + n \cdot \Delta^2 )$ & $O(n+\Delta)$ & Gribanov, Malyshev, Pardalos \& Veselov \cite{FPT_Grib} \\
    \hline
    DP & $\tilde O(n^{\omega} + m \cdot \Delta)$ & $O(n+\Delta)$ & {\color{red} this work} \\
    & & & $m = \min\{n, \Delta\}$ \\
    \hline
    \hline
    \end{tabular}
    
    \captionsetup{justification=centering}
    \caption{Complexity bounds for \ref{SVP_def} and \ref{CVP_def}}
    \label{SVP_tb}
\end{table}
\begin{remark}
Strictly speaking, the algorithms, presented in Theorems \ref{main_SVP_th} and \ref{main_CVP_th}, use of $O(n \cdot \Delta)$ space. But, they can be transformed to $O(\Delta)$-space algorithms without significant effort. Definitely, our algorithms use dynamic tables with $O(n \cdot \Delta)$ entries. It can be easily seen that if we want only to compute the optimal value of the objective function, then it is sufficient to store only one row of these tables at each computational step, which reduces the space requirement to $O(\Delta)$.

However, if we want to compute an optimal solution vector, then this simple trick is not applicable and more sophisticated technique need to be used. Such a technique is described, for example, in \cite[Paragraph~3.3]{KnapsackBook2004} (see also \cite{KnapsackDPTrick}), and it can be applied for our dynamic programming algorithms without any restrictions.
\end{remark}

\section{Data Structures}\label{data_str_sec}

In this Section, we describe data structures that will be used for our $(\min, +)$-convolution algorithms. The first two of them, the \emph{queue with minimum operation} and the \emph{compressed segment tree}, are classical. The third data structure is our modification of the segment tree, we call it the \emph{augmented segment tree}.  

\subsection{Queue with Minimum Support}\label{queue_subs}
The queue with minimum support is a classical data structure that looks to be folklore. We did not able to find any correct historical references on it. This data structure just represents a generic queue that stores elements of some linearly ordered set, but with an additional operation $min()$ that returns the current minimum. Let $Q$ be an instance of the queue, we list all the operations with their complexities:
\begin{enumerate}
    \item the operation $Q.min()$ returns a current minimum in $Q$. Its complexity is $O(1)$ in the worst case;
    \item the operation $Q.push(x)$ inserts an element $x$ to the tail of $Q$. Its complexity is $O(1)$ in the worst case;
    \item the operation $Q.pop()$ removes an element from the head of $Q$. Its complexity is $O(1)$ amortised.
\end{enumerate}

Since we are not able to give a correct reference to this data structure, we give a brief explanation of \emph{how it works}. First of all, note that it is easy to implement a stack with minimum support and with the worst-case complexity $O(1)$, for all the operations. To do this, we just need to create a second stack, which will store the current minimum. 

A queue $Q$ with minimum support can be implemented just by using two stacks $S_{h}$ and $S_{t}$ that are glued by the bottom side. The stack $S_{t}$ represents a tail of $Q$, and the stack $S_{h}$ represents a head of $Q$. Now, the operation $Q.min()$ can be implemented just by taking the minimum value between $S_{t}.min()$ and $S_{h}.min()$. When we need to insert a new element $x$ to $Q$, we just need to call $S_{t}.push(x)$. Finally, the operation $Q.pop()$ can be implemented in the following way: if $S_{h}$ is not empty, we just call $S_{h}.pop()$. If $S_{h}$ is empty, we move all the elements from $S_{t}$ to $S_{h}$ and call $S_{h}.pop()$ after that. Clearly, the worst case complexity of $Q.pop()$ is $O(n)$. But, since any element of $Q$ can be moved from $S_t$ to $S_h$ only ones, the amortised complexity of $Q.pop()$ is $O(1)$.

\subsection{Segment Tree}\label{segment_subs}

The segment tree is a classical data structure to perform the range minimum, the sum or update queries in sub-intervals of a given array, using only logarithmic worst-case time. We did not find any correct historical references on it, but a detailed description could be found in the internet \cite{SegmentTreeWeb}. The brief description of the weaker version without range update operations could be found in \cite[Section A.3]{LinearOperators}. Additionally, the work \cite{LinearOperators} gives a good survey and interesting new results about queries on arbitrary semigroups.

Let $T$ be an instance of a segment tree. We list the required operations with their complexities:
\begin{enumerate}
    \item the operation $T.build(A)$ builds the data structure on an array $A$ of length $n$. Its worst-case complexity is $O(n)$;
    \item the operation $T.min(i, j)$ returns a minimal element in the sub-array $A[i, j)$. Its worst-case complexity is $O(\log (n))$;
    \item the operation $T.add(i, j, x)$ adds the value of $x$ to all the elements of the sub-array $A[i, j)$. Its worst-case complexity is $O(\log (n))$. 
\end{enumerate}

Let us assume that $n = 2^d$. We call an interval $[i,j)$ \emph{basic}, if $[i,j) = [i \cdot 2^{d-k}, (i+1) \cdot 2^{d-k})$, for some $i \in \intint[0]{2^k - 1}$ and $k \in \intint[0]{d}$. The segment tree is represented as a full binary tree, where each node $v$ corresponds to some basic interval $[i_v, j_v)$ and additionally stores a minimum element in the sub-array $A[i_v, j_v)$. If $v$ is not leaf, then it has two children $a$ and $b$, corresponding to the intervals $[i_a, j_a) = [i_v, i_v + h)$ and $[i_v + h, j_v)$, where $h = (j_v-i_v)/2$. The leafs correspond to intervals of length $1$, associated with all the elements of $A$. If $v$ is a root node, we just have $[i_v, j_v) = [0, 2^d)$ and the minimum value in $v$ corresponds to the minimum in $A$.

The key idea that helps to compute the minimum value in a general interval $[i, j)$ is a special algorithm that splits a given interval $[i ,j)$ into at most $2 d$ basic intervals. We emphasise this in the following lemma, which will be used later.

\begin{lemma}\label{decomp_lm}
For any given interval $[i,j)$, there is an $O(d)$-complexity algorithm that splits $[i,j)$ into at most $2 d$ basic intervals, corresponding to the nodes of $T$. 
\end{lemma}

\subsection{Augmented segment tree}\label{aug_segment_subs}

Assume again that $n$ is a power of $2$. Let $A$ be an array of length $n$, and let $f(x) : [0,n) \to \RR$ be a function. Given $x \in \intint[0]{n-1}$ and $i,j \in \intint[0]{n}$, we are interested in the problem to efficiently compute a function 
\begin{multline*}
    \FC\bigl([i,j);\, x\bigr) = \\
    = \min\bigl\{ A[i] + f(x),\; A[i+1] + f(x+1),\; \dots,\; A\bigl[i +(j-i-1)\bigr] + f\bigl(x +(j-i-1) \bigr) \bigr\} = \\
    = \min\limits_{0 \leq k < j-i} \bigl\{ A[i + k] + f(x + k) \bigr\},
\end{multline*}
assuming that some preprocessing is done for $A$. To compute the values of $f$ for $x \in [0,n)$, the $\EVO$-oracle can be used (see Definition \ref{EVO_def}).

The following main property of the function $\FC\bigl([i,j);\, x\bigr)$ can be checked straightforwardly:
\begin{proposition}
Let an interval $[i,j)$ be partitioned into the intervals $[i_1, j_1)$ and $[i_2, j_2)$, i.e. $[i,j) = [i_1, j_1) \cup [i_2, j_2)$, where $i_2 \geq i_1$. Then,
\begin{equation}\label{F_property_eq}
    \FC\bigl([i,j);\, x\bigr) = \min\Bigl\{ \FC\bigl([i_1,j_1);\, x\bigr),\; \FC\bigl([i_2,j_2);\, x + j_1 - i_1\bigr) \Bigr\}.
\end{equation}
\end{proposition}

\begin{definition}
The interval $[a,b) \subseteq \RR$ is called \emph{integer} if its end-points $a$ and $b$ are integers.
\end{definition}

\begin{definition}
Let $f \colon [0,n) \to \RR$ and $\IC \subseteq [0,n)$ be an integer interval. Let $\IC_1, \IC_2, \dots, \IC_m$ be a set of integer intervals, such that
\begin{enumerate}
    \item the intervals $\IC_1, \dots, \IC_m$ partition $\IC$;
    \item for any $j$ and $x \in \IC_j$, we have either $f(x) \geq 0$, $f(x) \leq 0$, $f(x) > 0$ or $f(x) < 0$.
\end{enumerate}

A minimal set of such intervals $\IC_1, \dots, \IC_m$ is called a \emph{minimal sign partition of $f$}. The set of all such minimal partitions is denoted by $\PS(f, \IC)$.
\end{definition}

It is easy to see that a minimal sign partition is not unique.

\begin{definition}\label{SPO_def}
By $\SPO$ we denote the \emph{minimal sign partition oracle}. For given $f$ and $\IC \subseteq \dom(f)$, it returns some minimal sign partition from $\PS(f,\IC)$.
\end{definition}

Next, we need to define a special characteristic $p_f$ of a function $f \colon \IC \to \RR$, defined on an integer interval $\IC$, which will be extensively used further.
\begin{definition}
Let $f\colon \IC \to \RR$ be a function, defined on an integer interval $\IC$. Let us define a value $p_f$ in the following way. For $a \in \ZZ_{\geq 0}$ and $b \in \ZZ$, let us consider a function $g_{a b}(x) = f(x+a) - f(x) + b$. Let
\begin{gather*}
    p_f(a,b,\JC) = \max\left\{ \abs{\PC}\colon \PC \in \PS(g_{a b},\JC) \right\}\\
    p_f = \max\{p_f(a,b,\JC) \colon a \in \RR_{\geq 0},\, b \in \RR,\, \JC \subseteq \dom(g_{a b}) \}
\end{gather*}

In other words, $p_f$ is the maximal size of a minimal sign partition that $f(x + a) - f(x) + b$ can have on $\JC$, for the all possible values of $a,\, b$ and correct integer sub-intervals $\JC \subseteq \IC$.
\end{definition}

The following theorem defines the augmented segment tree data structure:
\begin{theorem}\label{main_atree_th}
Assume that $\EVO$ and $\SPO$ oracles are available. Let $f \colon [0,n) \to \RR$ be a function, $A$ be an array of length $n$, which is a power of $2$, and $p := p_f$.

There exists a data structure $T$, called the \emph{augmented segment tree}, that supports the following list of operations:
\begin{enumerate}
    \item the operation $T.build(A)$ builds the data structure for the array $A$ of length $n$. The worst-case $\SPO$-oracle and arithmetic complexities are $O(n^{\log_2 (p) + \frac{1}{1+\log_2 (p)}})$;
    \item the operation $T.query(i,j,x)$ returns the value of $\FC\bigl([i,j);\, x\bigr)$. The worst-case $\EVO$-oracle and arithmetic complexities are $O(\log^2(n))$.
\end{enumerate}
The data structure uses $O(n^{\log_2 (p) + \frac{1}{1+\log_2 (p)}})$ space. Calls to $\SPO$ oracle are performed for functions of the type $g(x) = f(x+a) - f(x+b) + c$, where $a,b \in \ZZ_{\geq 0}$, and $c \in \ZZ$. Calls to $\EVO$ are performed for $f$.
\end{theorem}
\begin{proof}

{\bf Description of data structure:} Our new data structure is a common segment tree $T$ with some additional augmentations. Here we will use the same notations, as in Subsection \ref{segment_subs}. 

We augment each vertex $v$ of $T$ with an additional data, represented by a finite set $\GS_v$ of functions $g : \DC_g \cap \ZZ \to \ZZ$, where each domain $\DC_g$ is an integer sub-interval of $[0,n)$ and $g$ acts on $\DC_g$ as a function $g(x) = A[j] + f(x + t)$, for some $j,t \in \intint[0]{n-1}$. 

We will support the following four invariants, for any $v \in T$:
\begin{itemize}
    \item Invariant 1: the intervals $\DC_g$, for any $g \in \GS_v$, must split $[0,n)$:
$$
[0,n) = \bigsqcup\limits_{g \in \GS_v} \DC_g;
$$

    \item Invariant 2: for any $x \in [0,n)$, there exists a unique function $g \in \GS_v$, such that $x \in \DC_g$, and 
    $$
    \FC\bigl([i_v, j_v); x\bigr) = g(x);
    $$
    
    \item Invariant 3: the functions $g \in \GS_v$ are stored in the sorted order with respect to the end-points of their domains $\DC_g$;
    
    \item Invariant 4: for any $g \in \GS_v$, the function $g(x)$ acts on $\DC_g$ like $g(x) = A[j] + f(x + t)$, for $j,t \in \intint[0]{n-1}$;  
\end{itemize}

{\bf Description and analysis of the query operation for basic intervals:} For the definition of \emph{basic intervals}, see Subsection \ref{segment_subs}. Assume that a vertex $v \in T$ is given, and we want to perform the $query(i_v, j_v, x)$ operation with respect to the basic interval $[i_v, j_v)$. Due to Invariant 2, we just need to find an appropriate function $g$ from the set $\GS_v$. Due to Invariant 3, the function $g$ can be found in $O(\log (n))$ time, because $\GS_v$ contains at most $n$ functions. Due to Invariant 4, $g(x)$ looks like $A[j] + f(x + t)$, so it can be computed, using a single call to $\EVO$. The total complexity is $O(\log_2(n))$.

{\bf Description and analysis of the query operation for general intervals:} Assume that an interval $[i,j)$ is given, and we want to perform the $query(i, j, x)$ operation. Due to Lemma \ref{decomp_lm}, there exist $m \leq 2 \log_2(n)$ vertices $v_1, v_2, \dots, v_m \in T$, such that $[i,j)$ is partitioned into the basic intervals $[i_{v_k}, j_{v_k})$, for $k \in \intint m$. Let us assume that $i_{v_1} < i_{v_2} < \dots < i_{v_m}$, and let $h_k = i_{v_k} - i_{v_1}$. Due to the property \eqref{F_property_eq}, we have
\begin{multline*}
    \FC\bigl([i,j);\, x\bigr) = \\
    = \min\Bigl\{\FC\bigl([i_{v_1},j_{v_1});\, x + h_1\bigr), \dots, \FC\bigl([i_{v_m},j_{v_m});\, x + h_m\bigr)\Bigr\} = \\
    = \min\limits_{k \in \intint m} \Bigl\{ \FC\bigl([i_{v_k},j_{v_k});\, x + h_k\bigr) \Bigr\}.
\end{multline*}

Consequently, due to the complexity bound on queries for basic intervals, the complexity of the $query(i,j,x)$ operation is $O(\log^2(n))$.

{\bf Description and analysis of the preprocessing:}

First of all, let us construct the standard segment tree $T$, described in Subsection \ref{segment_subs}, for the array $A$. It will take $O(n)$ time and space.

We need to show how to compute $\GS_v$, for any $v \in T$, and satisfy all the invariants. The algorithm is recursive: it starts from the leafs, and moves upper, until it meats the root of $T$. Let $v$ be a leaf. Since $j_v - i_v = 1$, $\FC\bigl([i_v,j_v); x\bigr) = A[i_v] + f(x)$. Consequently, $\GS_v$ consists of only one function $g(x) = A[i_v] + f(x)$, and $\DC_g = [0, n)$. 

Next, we assume that $v$ is not a leaf, and let $u$ and $w$ be the children of $v$. We will show how the set $\GS_v$ can be constructed from the sets $\GS_u$ and $\GS_w$, based on the formula 
\begin{equation}\label{basic_reduction_eq}
    \FC\bigl([i_v,j_v);\, x\bigr) = \min\Bigl\{
    \FC\bigl([i_u,j_u);\, x\bigr),\;
    \FC\bigl([i_w,j_w);\, x + j_u-i_u\bigr)
    \Bigr\},
\end{equation}
which is a direct application of \eqref{F_property_eq}.

Let $\PC_u$ and $\PC_w$ be the sets of end-points of intervals, representing the domains of functions inside $\GS_u$ and $\GS_w$. We assume that $0, n \in \PC_u$ and $0, n \in \PC_w$. Clearly, $\abs{\PC_u} = \abs{\GS_u} + 1$ and $\abs{\PC_w} = \abs{\GS_w} + 1$. Due to Invarinat 3, we can assume that $\PC_u$ and $\PC_w$ are sorted. Next, we merge $\PC_u$ and $\PC_w$ into $\PC_v$, maintaining the same sorting order, and remove the duplicates. The last step can be done in $O(\abs{\GS_u} + \abs{\GS_w})$-time, since $\PC_u$ and $\PC_w$ are sorted. Since the points $0,n$ are common for both $\PC_u$ and $\PC_v$, we have 
\begin{equation}\label{PCv_size}
    \abs{\PC_v} \leq \abs{\GS_u} + \abs{\GS_w}.
\end{equation}

Take a pair $\nu, \tau$ of consecutive points in $\PC_v$. Due to Invariant 2, there exist unique functions $g_u \in \GS_u$ and $g_w \in \GS_w$, such that $[\nu, \tau) \subseteq \DC_{g_u} \cap \DC_{g_w}$. Due to the formula \eqref{basic_reduction_eq}, for $x \in [\nu, \tau)$, we have 
$$
\FC\bigl([i_v,j_v);\, x\bigr) = \min\bigl\{
    g_u(x),\;
    g_w(x + j_u-i_u)
    \bigr\}.
$$

Let $h(x) = g_u(x) - g_w(x + j_u - i_u)$, defined on $[\nu, \tau)$. Due to Invariant 4, the function $h(x)$ has the form $f(x+a) - f(x+b) + c$, for some $a,b \in \ZZ_{\geq 0}$ and $c \in \ZZ$.

To efficiently precompute $\FC\bigl([i_v,j_v); x\bigr)$ for $x \in [\nu, \tau) \cap \ZZ$, we need to compute a minimal sign partition $\SC \in \PS(h, [\nu, \tau))$. It can be done by a single call to $\SPO$. Now, for any interval $\IC \in \SC$, if $h(x) \geq 0$ on $\IC$, then $\FC\bigl([i_v,j_v); x\bigr) = g_u(x)$ and $\FC\bigl([i_v,j_v); x\bigr) = g_w(x + j_u - i_u)$ in the opposite case $h(x) \leq 0$. Consequently, for any such interval $\IC$, we create a new function $g_{\IC}$ and put it inside $\GS_v$ in the sorted order with respect to endpoints of $\IC$. Hence, the interval $[\nu, \tau)$ will be decomposed into at most $p$ new sub-intervals, and the same number of new functions will be added into $\GS_v$.

Now, let us estimate the time and space requirements to build the set $\GS_v$. As it was shown before, for any pair $[\nu,\tau)$ of consecutive points from $\PC_v$, we add at most $p$ functions to $\GS_v$.
Therefore, due to \eqref{PCv_size}, we have 
\begin{equation*}
    \abs{\GS_v} \leq \bigl(\abs{\GS_u} + \abs{\GS_w} - 1\bigr) \cdot p.
\end{equation*}
Denote $N(m) = \max\bigl\{\abs{\GS_v} \colon v \in T,\; j_v - i_v = m \bigr\}$, for $m \leq n$ being a power of $2$. Since $N(1) = 1$, we have 
\begin{equation*}
    N(m) \leq 2 \cdot N(m/2)\cdot p \leq (2p)^{\log_2 (m)} = m^{1 + \log_2 (p)}.  
\end{equation*}
And, since we always work in the interval $[0,n)$, 
\begin{equation}\label{nodes_num_eq}
    N(m) \leq \min\{m^{1 + \log_2 (p)},\; n\}.
\end{equation}

By analogy with $N(m)$, let us denote the maximal time to construct $\GS_v$ (in the assumption that $\GS_u$ and $\GS_w$ are already constructed) by  $t_{node}(m)$, where $m = j_v - i_v$. By the word "time", we mean both arithmetical and oracle complexities. Clearly, the definition is correct, because the value of $m$ is the same for all the vertices of the same level in $T$. Since the complexity to compute $\GS_v$ is linear with respect to the resulting size of $\GS_v$, due to \eqref{nodes_num_eq},
\begin{equation}\label{tnode_eq}
t_{node}(m) = O\bigl( N(m) \bigr)= O\bigl( \min\{m^{1 + \log_2 (p)},\; n\} \bigr).    
\end{equation}
Note additionally that the space requirements to store $\GS_v$ with the whole information, related to $v$, can be described by the same function $t_{node}(m)$.

Now, let us compute the total time and space complexity to construct the final augmented tree $T$. It can be expressed by the function 
$$
t(n) = \sum\limits_{k = 0}^{\log_2 (n)} 2^k \cdot t_{node}(n/2^k).
$$

Let $s = \Bigl\lceil \log_2\bigl( n^{\frac{1}{1+\log_2 (p)}} \bigr) \Bigr\rceil$. To calculate the asymptotic of $t(n)$, we split the sum into two parts and estimate elements of each sum, using \eqref{tnode_eq}:
\begin{multline*}
    t(n) \lesssim \sum\limits_{k = 0}^{s} 2^k \cdot n + \sum\limits_{k = s+1}^{\log_2 (n)} 2^k \cdot (n/2^k)^{1 + \log_2 (p)} \lesssim \\ 
    \lesssim n^{1 + \frac{1}{1 + \log_2(p)}} + n^{1 + \log_2 (p)} \cdot \sum\limits_{k = s+1}^{\log_2 (n)} 2^{- k \cdot \log_2 (p)}.
\end{multline*}

Estimating the sum at the end of the last formula, we have:
\begin{multline*}
    \sum\limits_{k = s+1}^{\log_2 (n)} 2^{- k \cdot \log_2 (p)} = \sum\limits_{k = s+1}^{\log_2 (n)} p^{- k}
    = \frac{p}{p-1} \cdot \Bigl( \frac{1}{p^{1+s}} - \frac{1}{p^{1+\log_2 (n)}} \Bigr) = \\
    = \frac{1}{p-1} \cdot \Bigl( \frac{1}{p^{s}} - \frac{1}{p^{\log_2 (n)}} \Bigr) \leq 
    \frac{1}{p-1} \cdot \Bigl( \frac{1}{n^{\frac{\log_2 (p)}{1+\log_2 (p)}}} - \frac{1}{n^{\log_2 (p)}} \Bigr).
\end{multline*}

Finally, the total time and space requirements can be estimated as follows 
\begin{multline*}
    t(n) \lesssim n^{1 + \frac{1}{1 + \log_2 (p)}} + \frac{1}{p} \cdot n^{1 + \log_2 (p) - \frac{\log_2 (p)}{1+\log_2 (p)}} = \\ 
    = n^{1 + \frac{1}{1+ \log_2 (p)}} \cdot \Bigl( 1 + \frac{1}{p} \cdot n^{-1 + \log_2 (p)} \Bigr) = \\
    = O\Bigl( n^{\log_2 (p) + \frac{1}{1 + \log_2 (p)} } \Bigr).
\end{multline*}

Theorem \ref{main_atree_th} is proved. $\Box$
\end{proof}

\section{Structured $(\min,+)$-convolution algorithms}\label{conv_sec}

In this Section, we describe how to solve the problem \ref{min_dot_def} for the following cases:
\begin{enumerate}
    \item $f$ is linear: $f(x) = \alpha \cdot x + \beta$;
    
    \item $f$ is piece-wise linear: $f(x)$ is represented by a piece-wise linear function with $p$ pieces;
    
    \item $f \in \ZZ^d[x]$;
    
    \item $f$ is concave. 
\end{enumerate}

\subsection{The linear case}\label{linear_conv_subs}

W.l.o.g. we can assume that $b_i = \alpha \cdot i$, for some $\alpha \in \QQ$. We will use a queue $Q$, which was described in Subsection \ref{queue_subs}.
The algorithm consists of $m-n$ steps: 

At the first step, we just initialise $Q$ with the elements 
$$
a_0 + b_0,\; a_1 + b_1,\; \dots,\; a_{m-1} + b_{m-1},
$$ which can be done in $O(m)$-time. After that, we assign $c_0 := Q.min()$, which can be done in $O(1)$-time.

Note that the difference between elements in $Q$ is exactly $\alpha$. We will support the following invariant:
$$
\boxed{
\begin{gathered}
\text{after the $k$-th step the queue $Q$ contains the following elements:}\\
a_k + \alpha \cdot k,\; a_{k+1} + \alpha \cdot (k+1),\; \dots,\; a_{k+m-1} + \alpha \cdot (k + m-1).
\end{gathered}
}
$$

Assuming that the $k$-th step has been done and $c_k$ has been computed, let us show how to perform the $(k+1)$-th step. We call the $Q.pop()$ and after that call $Q.push(x)$, for $x = a_{k+m} + \alpha \cdot (k+m)$. The last operations will satisfy the invariant at the $(k+1)$-th step. Now, we can put $c_{k+1} := Q.min() - \alpha\cdot(k+1)$, due to the invariant, it is the correct value of $c_{k+1}$.

Since the amortised complexity of each step is $O(1)$, the total arithmetical complexity bound is $O\bigl(m + (n-m)\bigr) = O(n)$.

\subsection{The piece-wise linear case}\label{plinear_subs}

W.l.o.g. we can assume that $f(x)$ is defined on $[0,m)$ by the following three vectors: $\alpha, \beta \in \QQ^p$, and $u \in \QQ_{\geq 0}^{p+1}$. We assume, that $u_0 = 0$, $u_{p} = n$, and $u_{j-1} < u_j$, for $j \in \intint{p}$. The formula for $f$ is:
$$
f(x) = \beta_k + \alpha_k \cdot x, \quad \text{for $x \in [u_{k-1}, u_k)$}.
$$

Assuming $b_i = f(i)$, let us show how to compute the elements of $c$. We will use the compressed segment tree data structure $T$, described in Subsection \ref{segment_subs}. The algorithm consists of $n-m$ steps:

At the first step, we construct the array $A := a$ and assign $A[i] := A[i] + b_i$, for all $i \in \intint[0]{m-1}$. It takes $O(n)$ arithmetic operations. Next, we initialise $T$, by calling $T.build(A)$. It also takes $O(n)$-time. 

We will support the following invariant: 
$$
\boxed{
\begin{gathered}
\text{after the $k$-th step has been done, the sub-array $A[k, k+m)$ consists of the elements}\\
a_k + b_0,\; a_{k+1} + b_1,\; \dots,\; a_{k+m-1} + b_{m-1}.
\end{gathered}
}
$$
Consequently, the equality $T.query(k, k+m) = c_k$ holds after the $k$-th step has been finished.

Now, let us show how to perform the $(k+1)$-th step with the complexity $O(p \cdot \log (n))$. Fix a number $j \in \intint{p}$ and consider the sub-array $A[k+1+u_{j-1}, k+ 1 + u_{j})$. Let $d_j = u_j - u_{j-1}$. By the invariant, the first $d_{j} - 1$ elements of this array are equal to 
$$ 
a_{k+1 + u_{j-1}} + \beta_j + \alpha_{j} \cdot u_{j-1},\; \dots,\; a_{k+u_{j}-1} + \beta_j + (u_{j}-1) \cdot \alpha_j.
$$ The last element $A[k + u_{j}]$ is equal to $a_{k+u_{j}} + f(u_{j}) = a_{k+u_{j}} + \beta_{j+1} + \alpha_{j+1}\cdot u_{j}$. Consequently, to make the first $d_j -1$ elements of the sub-array to satisfy the invariant, we need to make the $update(k+1+u_{j-1}, k+ u_{j}, -\alpha_j)$ operation, which can be done in $O(\log(n))$-time. Since $k+u_j = (k+1) + (u_j - 1)$, the last element $A[k + u_{j}]$ must be assigned to $a_{k+u_{j}} + f(u_{j}-1)$, which can also be done in $O(\log(n))$-time. After applying this procedure for all $k \in \intint{p}$, the invariant for the $k+1$-th step will be satisfied, and the value of $c_{k+1}$ can be computed just by using the $query(k+1, k+1+m)$ operation. The complexity of the considered step is $O(p \cdot \log(n))$.

The total arithmetic complexity of the algorithm is $O\bigl(n + (n-m) \cdot p \cdot \log(n) \bigr) = O\bigl(p \cdot n \cdot \log(n)\bigr)$.

\subsection{The polynomial and concave cases}\label{poly_subs}

Let us consider a function $f \colon [0,m) \to \RR$. We assume that $\EVO$ and $\SPO$ oracles are supported.

Let us estimate the oracle and arithmetic complexities to solve  \ref{min_dot_def} with $b_i = f(i)$. We will use the augmented segment tree data structure, described in Subsection \ref{aug_segment_subs}. Let us assume that $n$ be a power of $2$ and choose a block size $B$, which is also be a power of $2$. 

We assign $A := a$ and split $A$ into $n/B$ consecutive blocks of size $B$. Let $\BC_j$ be the interval representing the $j$-th block, i.e. $\BC_j = [(j-1)\cdot B, j \cdot B)$. Now, for each $j \in \intint{n/B}$, we construct the augmented segment tree data structure $T_j$ by calling the operation $T_j.build(A[\BC_j])$. Due to Theorem \ref{main_atree_th}, the oracle, arithmetic, and space complexities of this step can be expressed by the formula
\begin{equation}\label{block_prep_complexity_eq}
    O\Bigl( \frac{n}{B} \cdot B^{\log_2 (p) + \frac{1}{1 + \log_2 (p)}} \Bigr) 
    = O\Bigl( n \cdot B^{\sigma - 1} \Bigr),
\end{equation} where $\sigma = \log_2 (p) + \frac{1}{1 + \log_2 (p)}$.

The algorithm consist of $n-m$ steps. At the $k$-th step we try to compute the value of $c_k$, using the hint that $c_k = \FC\bigl([k, k+m);\,0\bigr)$.

Let $\BC_{j}, \dots, \BC_{j + t}$ be the consecutive blocks, affected by the window $[k, k+m)$, where $t = O(m/B)$. If $t = 1$, then we can just put $c_k := T_{j}.query(\nu, \tau, 0)$, where $\nu = k-B\cdot (j-1)$ and $\tau = k-B\cdot (j-1) + m$ are the relative coordinates of the $[k, k+m)$ window in $\BC_{j}$.

Now, let us consider the case $t \geq 2$. Let $s = j \cdot B - k$ be the size of the intersection of $\BC_j$ with $[k, k+m)$. Using Proposition \eqref{basic_reduction_eq}, we have
\begin{multline*}
    c_k := \FC\bigl([k,k+m);\, 0\bigr) = \min\Bigl\{ \\
    \FC\bigl(\BC_j \cap [k,k+m);\, 0\bigr),\; \FC\bigl(\BC_{j+1};\, s \bigr),\; \FC\bigl(\BC_{j+2};\, s + B \bigr),\;\dots \\
    \FC\bigl(\BC_{j+t-1};\, s + B \cdot (t-2) \bigr),\; \FC\bigl(\BC_{j+t} \cap [k,k+m);\, s + B \cdot (t-1) \bigr)
    \Bigr\}
\end{multline*}

The interval $\BC_j \cap [k,k+m)$ is a suffix of $\BC_j$. So, the first value $$\FC\bigl(\BC_j \cap [k,k+m);\, 0\bigr)$$ can be computed by a call to $T_{j}.query\bigl(B-s,\, B,\, 0\bigr)$. 

The interval $\BC_{j + t} \cap [k,k+m)$ is a prefix of $\BC_{j+t}$. So, the last value $$\FC\bigl(\BC_{j + t} \cap [k,k+m);\, s + B \cdot (t-1)\bigr)$$ can be computed by a call to $$T_{j+t}.query\bigl(0,\, k+m - B\cdot(t-1),\, s + B \cdot (t-1)\bigr).$$ Here $k+m - B\cdot(t-1)$ is the size of $\BC_{j + t} \cap [k,k+m)$. 

The intermediate values $\FC\bigl(\BC_{j + i};\, s + B \cdot (i-1)\bigr)$, for $i \in \intint[1]{t-1}$, can be computed by calls to $T_{j+i}.query(0,\,B,\,s + B \cdot (i-1))$.

Therefore, $c_k$ can be computed as the minimum between all this values. The respective arithmetic and oracle complexity is expressed by 
\begin{equation*}\label{block_k_complexity_eq}
O\bigl(\frac{m}{B} \cdot \log^2 (B)\bigr).
\end{equation*}

Consequently, the total complexity of $m-n$ steps without the initial preprocessing can be estimated as
$$
O\Bigl( (n-m) \cdot \frac{m}{B} \cdot \log^2(m) \Bigr) = O\Bigl( \frac{n^2}{B} \cdot \log^2(n) \Bigr).
$$

Now, the total algorithm complexity (together with the preprocessing, see the formula \eqref{block_prep_complexity_eq}) can be expressed by the formula
\begin{equation*}
    O\Bigl( n \cdot B^{\sigma - 1} + \frac{n^2}{B} \cdot \log^2(n) \Bigr).
\end{equation*}

Actually, a more detailed formula holds
\begin{equation}\label{block_full_complexity_eq}
    O\Bigl( n \cdot B^{\sigma - 1} \cdot T_{\SPO} + \frac{n^2}{B} \cdot \log^2(n) \cdot T_{\EVO} \Bigr).
\end{equation}

We will try to balance this formula, solving the equation 
\begin{equation*}
    n \cdot B^{\sigma - 1} = \frac{n^2}{B}.
\end{equation*}
Clearly, it is equivalent to
\begin{equation*}
    B^{\sigma} = n.
\end{equation*}
So, the $B$ parameter could be chosen as $B = n^{\frac{1}{\sigma}}$. 
After this substitution, the total complexity becomes
\begin{equation}\label{gen_minconv_complexity_eq}
O\Bigl( \bigl( T_{\SPO} + \log^2(n) \cdot T_{\EVO} \bigr) \cdot n^{1 + \frac{\sigma-1}{\sigma}} \Bigr),    
\end{equation}

Due to \eqref{block_prep_complexity_eq}, the space complexity is
\begin{equation}\label{gen_minconv_space_eq}
    O\Bigl(n^{1 + \frac{\sigma-1}{\sigma}} \Bigr).
\end{equation}

Now, let us consider the polynomial case.
\begin{theorem}\label{minconv_poly_th}
Let $f \in \ZZ^d[x]$. Then, \ref{min_dot_def} can be solved by an algorithm with the arithmetic complexity bound
$$
O\bigl( d^3 \cdot n^{1 + \frac{\sigma-1}{\sigma}} \cdot \log^2(n) \bigr),
$$ where $\sigma = \log_2(d) + \frac{1}{1+\log_2(d)}$.
\end{theorem}
\begin{proof}
Clearly, the complexity of $\EVO$ for $f$ is $O(d)$. 

Let us estimate the complexity of $\SPO$ for polynomials of the type $g(x) = f(x+a) + b - f(x)$. Since $\deg(g) \leq d-1$, the size $p$ of any minimal sign partition $g$ is bounded by $d$. To calculate this partition on a given interval $\IC$, we clearly need to localize all the roots of $g$ inside $\IC$. Since intervals in the resulting sign partition need to have integer end-points, we do not need to compute the roots exactly. Instead of that, we simply can calculate them with the additive accuracy $1/3$ and round-off to the nearest integer.

To localize the roots of $g$ on $\IC$, we will use the classical Budan--Fourier theorem. It states that for any interval $(\nu,\tau)$ with $g(\nu) \not= 0$ and $g(\tau) \not= 0$ the number of roots of $g(x)$ in the interval $(\nu, \tau)$ is equal or less than the value of $W(\nu) - W(\tau)$, where $W(x)$ is the number of sign changes in the sequence $f(x),\, f'(x),\, f''(x), \dots$. Note that the Budan--Fourier theorem does not calculate the exact number of roots, the real number of roots in $[\nu,\tau)$ can be less by an even number. But, in our case, we only need to know how the sign changes, when $x$ crosses an integer landmark point. So, this method can be used without restrictions.

Clearly, the sequence of $g(x)$-derivatives can be computed, using $O(d^2)$ arithmetic operations. After that, given a point $x$, the value of $W(x)$ could also be computed, using $O(d^2)$ arithmetic operations. Then, using the standard dichotomy principle, we could localize all the roots of $g$ on $\IC$ with the additive accuracy $1/3$, using $O(d \cdot \log (n))$ calls to $W(x)$. Hence, the complexity of $\SPO$ on $g$ is $O(d^3 \cdot \log (n) )$. 

Let us consider the formula \eqref{gen_minconv_complexity_eq}. Since $\sigma = \sigma(p)$ is a monotone function, we have $\sigma(p) \leq \sigma(d)$. Then, together with the complexity bounds for $\SPO$ and $\EVO$, the formula \eqref{gen_minconv_complexity_eq} gives the desired complexity bound for \ref{min_dot_def}. $\Box$
\end{proof}

\medskip

Now, we are going to consider the concave case. First of all, we need some auxiliary lemmas.




\begin{lemma}\label{convex_ext_lm}
Let $(x_0,f_0),\, (x_1,f_1),\, \dots, (x_{n-1},f_{n-1})$ be a sequence of pairs from $\RR^2$. By $d_i$ we denote $(f_i - f_{i-1})/(x_i - x_{i-1})$. Assume that 
\begin{gather*}
    x_0 \leq x_1 \leq \dots \leq x_{n-1} \\
    d_1 \leq \dots \leq d_{n-1}.
\end{gather*}

Then, there exists a $C^1$-smooth convex function $f \colon \RR \to \RR$, such that $f(x_i) = f_i$.
\end{lemma}
\begin{proof}
Put $\IC = \intint[0]{n-1}$. To prove the lemma, we will use the criteria, given in \cite[Corollary~1]{SmoothConvexInterpol2017}. For given $g_0, g_1, \dots, g_{n-1}$, it states that the $C^{1,L}$-smooth convex function $f$ with $f'(x_i) = g_i$ exists if and only if the following conditions are satisfied:
\begin{equation}\label{conv_smooth_criteria}
    f_i \geq f_j + g_j\cdot(x_i-x_j) + \frac{1}{2L} \cdot \abs{g_i - g_j},\quad \text{for } i,j \in \IC.
\end{equation}

We construct $g_i$ in the following way. We choose $g_0 < d_1$ and $g_{n-1} > d_{n-1}$. For $i \in \intint[1]{n-2}$, if $d_{i} < d_{i+1}$, we choose $g_i$ strictly between $d_{i}$ and $d_{i+1}$: $d_{i} < g_i < d_{i+1}$. In the opposite case, when $d_{i} = d_{i+1}$, we just set $g_i := d_i$.

Fix $i,j \in \intint[1]{n-2}$. Assume firstly that $d_{k} = d_{l}$, for all $k,l$ between (inclusively) $i,j$. Then, the following equality, for any $\varepsilon > 0$, holds, which follows from the definition of $d_i$:
$$
f_i = f_j + g_j \cdot (x_i - x_j) = f_j + g_j \cdot (x_i - x_j) + \varepsilon \cdot \abs{g_i - g_j}.
$$

Now, assume that $d_k \not= d_{k+1}$, for some $k$ between (inclusively) $i,j$.  Then, since $g_k > d_k$, we have
$$
f_i > f_j + g_j \cdot (x_i - x_j),
$$ and consequently 
$$
f_i \geq f_j + g_j \cdot (x_i - x_j) + \varepsilon \cdot \abs{g_i - g_j},
$$ for any sufficiently small $\varepsilon > 0$.
Finally, if $i = 0$ or $j = n-1$, the same inequality holds, because $g_0 < d_1$ and $g_{n-1} > d_{n-1}$.

Therefore, since $\IC$ is finite, we can choose $\varepsilon$ sufficiently small, such that the following inequality will hold, for any $i,j \in \IC$:
$$
f_i \geq f_j + g_j \cdot (x_i - x_j) + \varepsilon \cdot \abs{g_i - g_j}.
$$
The last fact satisfies the conditions \eqref{conv_smooth_criteria} with $L = 1/(2\varepsilon)$, which proves the lemma. $\Box$


\end{proof}

\begin{lemma}\label{sign_part_lm}
Let $f \colon [\alpha,\beta) \to \RR$ be a convex/concave function ($\alpha, \beta \in \RR$). The following statements hold:
\begin{enumerate}
    \item $p_f \leq 2$;
    \item Assume that $\EVO$ oracle is supported for $f$. Let $g(x) = f(x+a) + b - f(x)$, for some $a,b \in \RR$. Then, given a bounded integer interval $\IC$ of length $n$, some minimal sign partition from $\PS(g,\IC)$ can be computed, using $O(\log(n))$ calls to $\EVO$.
\end{enumerate}
\end{lemma}
\begin{proof}
Assume that $f$ is convex. In the opposite case, we can consider a function $$-g(x) = \bigl(-f(x+a)\bigr) - b -\bigl(-f(x)\bigr), b \in \RR.$$ Clearly, all the sign partitions of $g(x)$ and $-g(x)$ are equivalent. 

Next, we can assume that $f$ is $C^1$-smooth. Definitely, if $f$ is not $C^1$-smooth, due to Lemma \ref{convex_ext_lm}, there exists a convex $C^1$-smooth function $h$, such that $h(x) = f(x)$, for all $x \in [\alpha,\beta) \cap \ZZ$. Since any minimal sign partition of $f$ consists of intervals with integer end-points, the sign partitions for $h$ and $f$ are equivalent, and we can use $h$ instead of $f$. 

Additionally, we assume that $a > 0$, because, in the opposite case, both statements are trivial. Now, we claim that the equality $g(x) = 0$ is possible only for points in some connected interval $[\nu, \tau] \subseteq [\alpha,\beta)$. Assume that $g(\nu) = 0$ and $g(\tau) = 0$, for some $\nu,\tau \in [\alpha,\beta)$ with $\tau > \nu$.
By definition, we have
\begin{gather*}
f(\nu+a) - f(\nu) = -b,\\
f(\tau+a) - f(\tau) = -b.
\end{gather*}
Due to $C^1$-smoothness of $f$, we can use the fundamental theorem of calculus:
\begin{gather}
    \int\limits_{\nu}^{\nu+a} f'(x) dx = -b,\label{nu_int_eq}\\
     \int\limits_{\tau}^{\tau+a} f'(x) dx = -b.\label{tau_int_eq}
\end{gather}

Put $\delta = \tau - \nu$. Make a change of variables $x \to x - \delta$ in \eqref{nu_int_eq} and $x \to x + \delta$ in \eqref{tau_int_eq}:
\begin{gather}
\int\limits_{\tau}^{\tau+a} f'(x - \delta) dx = -b,\label{nu_int_mapped_eq}\\
\int\limits_{\nu}^{\nu+a} f'(x + \delta) dx = -b.\label{tau_int_mapped_eq}    
\end{gather}

Subtracting \eqref{nu_int_eq} from \eqref{tau_int_mapped_eq} and \eqref{nu_int_mapped_eq} from \eqref{tau_int_eq}, we get
\begin{gather*}
    \int\limits_{\nu}^{\nu+a} \bigl(f'(x+\delta) - f'(x)\bigr) dx = 0,\\
    \int\limits_{\tau}^{\tau+a} \bigl( f'(x) - f'(x-\delta) \bigr) dx = 0.
\end{gather*}
Since $f$ is convex, $f'(x+\delta) - f'(x) \geq 0$ and $f'(x) - f'(x-\delta) \geq 0$, and, consequently,
\begin{gather*}
    \forall x \in [\nu,\nu+a] \quad f'(x) = \const,\\
    \forall x \in [\tau,\tau+a] \quad f'(x) = \const.
\end{gather*}

Again, since $f'(x)$ is convex, $f'(x) = \const$, for the whole interval $[\nu, \tau]$. Consequently, $g'(x) = 0$ and $g(x) = \const$, for $x \in [\nu, \tau]$. Since $g(\nu) = 0$, it holds that $g(x) = 0$, for $x \in [\nu, \tau]$. So, the claim is proved.

Therefore, any given interval $\IC \subseteq [\alpha,\beta)$, where $g$ is well-defined, can be partitioned into at most three parts: strict inequalities $g(x) > 0$ or $g(x) < 0$ on the left and right sides, and equality $g(x) = 0$ in the middle. Consequently, any minimal sign partition of $g$ on $\IC$ consists of at most $2$ pieces, which proves the inequality $p_f \leq 2$.

To calculate such a partition, we need to find an integer point $z \in \IC$ with $g(z) = 0$. Or, in the case, when $g(x) \not= 0$ for all $x \in \IC \cap \ZZ$, we need to calculate a point $z$, such that $g(z) < 0$ and $g(z+1) > 0$, or vice-versa. Clearly, in both cases we can use the standard dichotomy principle that takes $O( \log(n) )$ calls to $\EVO$. $\Box$


\end{proof}

\begin{theorem}\label{minconv_convex_th}
Let $f\colon [0,m) \to \RR$ be a concave function. Assume that $\EVO$ oracle is available.

Then, \ref{min_dot_def} can be solved by an algorithm with the arithmetic and oracle complexity, bounded by
$$
O\bigl( n^{4/3} \cdot \log^2(n) \bigr).
$$
\end{theorem}
\begin{proof}
Let $g(x) = f(x+a) + b - f(x)$ and $\IC \subseteq [0,n)$. Due to Lemma \ref{sign_part_lm}, the complexity of $\SPO$-oracle with the input pair $(g,\IC)$ is bounded by $O(\log(n))$ calls to $\EVO$. Additionally, $p := p_f \leq 2$.

Let us consider the formula \eqref{gen_minconv_complexity_eq}. Since $\sigma = \sigma(p)$ is a monotone function, we have $\sigma(p) \leq \sigma(2) = 3/2$. Together with the complexity bound for $\SPO$, the formula \eqref{gen_minconv_complexity_eq} gives the desired complexity bound for \ref{min_dot_def}. The proof is finished. $\Box$
\end{proof}

\section{Applications for the bounded knapsack}\label{knap_sec}

Let us consider \ref{knap_gen_def}. W.l.o.g. we can assume that $u_k \leq \lfloor W/w_k \rfloor$, for $k \in \intint n$. Additionally, we consider the minimization problem instead of maximization, since we can work with $-f$ instead $f$, which preserves the separability property. 

Let us consider a very basic dynamic program, dated to Bellman \cite{Bellman}. For $k \in \intint[1]{n}$ and $w_0 \in \intint[0]{W}$, by $\DPR(k,w_0)$ we denote the following problem:
\begin{gather}
    \sum\limits_{i = 1}^{k} f_i(x_i) \to \min\notag\\
    \begin{cases}
    \sum\limits_{i = 1}^k w_i x_i = w_0\\
    0 \leq x_i \leq u_i\\
    x \in \ZZ^k.
    \end{cases}\label{knap_DP_prob}
\end{gather}

Clearly, the problem $\DPR(n,W)$ is equivalent to the original problem \ref{knap_gen_def}. For $2 \leq k \leq n$, the value $\DPR(k,w_0)$ can be computed, using the values $\DPR(k-1,\cdot)$ by the following formula:
\begin{equation}\label{knap_DP_rec}
\DPR(k,w_0) = \min\limits_{j \in \intint[0]{u_k}}\Bigl\{ \DPR(k-1, w_0 - w_k \cdot j) + f_k(j) \Bigr\},
\end{equation}
assuming that $\DPR(k,w_0) = +\infty$ for $w_0 < 0$.

For $k = 1$, we have 
\begin{equation}\label{knap_DP_1}
\DPR(1,w_0) = \begin{cases}
f_1(w_0/w_1) \quad \text{if }w_1 \mid w_0\text{ and }0\leq w_0/w_1 \leq u_1\\
+\infty \quad \text{, in the opposite case}.
\end{cases}
\end{equation}

Let us estimate the complexity to compute all the values $\DPR(k,w_0)$, for $k \in \intint{n}$ and $w_0 \in \intint[0]{W}$.

Clearly, the values $\DPR(1,w_0)$ can be computed with $O(W)$ operations. Definitely, each of the values $\DPR(1,0),\,\DPR(1,w_1),\,\DPR(1,2 w_1),\,\dots$ can be computed with $O(1)$ operations, using the formula \eqref{knap_DP_1}. For other values of $w_0$, we just set $\DPR(1,w_0) = +\infty$.

Now, let us show how the computation of $\DPR(k,w_0)$ can be reduced to $(\min,+)$-convolution. Fix $k \geq 2$ and some residue $r$ modulo $w_k$. We define the sequences $\{a_i\}_{i \in \intint[0]{u_k}}$, $\{b_i\}_{i \in \intint[0]{u_k}}$, and $\{c_i\}_{i \in \intint[0]{u_k}}$ as follows:
\begin{gather*}
    a_i = \DPR(k-1,r + i \cdot w_k),\\
    b_i = f_k(i),\\
    c_i = \DPR(k,r + i \cdot w_k).
\end{gather*}

Assuming that $a_i = b_i = c_i = 0$, for $i < 0$, and due to \eqref{knap_DP_rec}, we have
\begin{equation}\label{knap_a_rec}
c_i = \min\limits_{j \in \intint[0]{i}}\bigl\{ a_{i-j} + b_j \bigr\}.
\end{equation}
That gives $c = (a \star b)[0,u_k]$. Therefore, considering all the values of $r$, the complexity to compute the level $\DPR(k,\cdot)$, in the assumption that the level $\DPR(k-1, \cdot)$ has already been computed, can be expressed by 
$$
O\bigl(w_k \cdot T_{conv}(u_k)\bigr),
$$ where $T_{conv}(\cdot)$ denotes the complexity of the $(\min,+)$-convolution. The total complexity of the whole dynamic programming scheme is
\begin{equation}\label{knap_total_DP_compl}
    O\Bigl( W + \sum\limits_{k=2}^n w_k \cdot T_{conv}(u_k) \Bigr).
\end{equation}

Using the previous formula, the inequality $u_i \leq \lfloor W/w_k \rfloor$, and different $T_{conv}$-complexity results, due to Subsection \ref{plinear_subs} and Theorems \ref{minconv_poly_th}, \ref{minconv_convex_th} of Subsection \ref{poly_subs}, we obtain the following result:

\begin{theorem}\label{main_knap_comp_th}
The following statements hold:
\begin{enumerate}
    \item The problem \ref{knap_plin_def} can be solved by $O\bigl( p \cdot n \cdot W \cdot \log(W)\bigr)$ arithmetic complexity algorithm;
    
    \item The problem \ref{knap_conv_def} can be solved by an algorithm with the arithmetic complexity bound
    \begin{multline*}
    O\Bigl( W^{4/3} \cdot \bigl( w_2^{-1/3} + w_2^{-1/3} + \dots + w_n^{-1/3} \bigr) \cdot \log^2(W) \Bigr) = \\
    = O\bigl( n \cdot W^{4/3} \cdot \log^2(W) \bigr);
    \end{multline*}
    
    \item The problem \ref{knap_poly_def} can be solved by an algorithm with the arithmetic complexity bound
    \begin{multline*}
    O\Bigl( d^3 \cdot W^{1+\frac{\sigma-1}{\sigma}} \cdot \bigl( w_2^{\frac{1-\sigma}{\sigma}} + w_2^{\frac{1-\sigma}{\sigma}} + \dots + w_n^{\frac{1-\sigma}{\sigma}} \bigr) \cdot \log^2(W) \Bigr) = \\
    = O\bigl( d^3 \cdot n \cdot W^{1+\frac{\sigma-1}{\sigma}} \cdot \log^2(W) \bigr),
    \end{multline*}
    where $\sigma = \log_2(d) + \frac{1}{1 + \log_2(d)} \geq 1$.
\end{enumerate}

All computations are performed with integer numbers of the size $O(\log(W))$.
\end{theorem}

\section{\ref{SVP_def} and \ref{CVP_def} dynamic programming algorithms}\label{SVP_CVP_sec}

\subsection{\ref{SVP_def} problem}\label{SVP_subs}

Let us consider the generalized problem \ref{gen_SVP_def}.



It is a known fact (see, for example, \cite{Schrijver,SNFOptAlg,Zhendong}) that there exist unimodular matrices $P \in \ZZ^{n \times n}$ and $Q \in \ZZ^{n \times n}$, such that $P A Q = S$, where $S \in \ZZ_{\geq 0}^{n \times n}$ is a diagonal non-degenerate matrix. Moreover, $\prod_{i = 1}^{k} S_{ii} = \Delta_{\gcd}(A,k)$, for any $k \in \intint{n}$, and, consequently, $S_{ii} \mid S_{(i+1) (i+1)}$, for $i \in \intint{n-1}$. Here $\Delta_{\gcd}(A,k)$ denotes the greatest common divisor of $k \times k$ sub-determinants of $A$. The matrix $S$ is called the \emph {Smith Normal Form} (or, shortly, the SNF) of $A$.

Using the SNF, we can reformulate the problem \eqref{gen_SVP_def}:
\begin{gather}
    \sum\limits_{i=1}^n f(x_i) \to \min\notag\\
    \begin{cases}
    P x \equiv \BZero \pmod{S \cdot \ZZ^n}\\
    x \in \ZZ^n \setminus \{\BZero\}.
    \end{cases}\label{SNV_SVP_def}
\end{gather}

Let us consider the quotient group $\GC = \ZZ^n/S\cdot\ZZ^n$ (with respect to addition in $\ZZ^n$), and define $g_i = P_i \bmod \diag(S)$, where $P_i$ is the $i$-th column of $P$ and $i \in \intint n$. We identify the vectors $g_i$ with the elements of the group $\GC$. Clearly, $\abs{\GC} = \Delta$. Then, the problem \eqref{SNV_SVP_def} can be reformulated as a minimization problem on $\GC$:
\begin{gather}
    \sum\limits_{i=1}^n f(x_i) \to \min\notag\\
    \begin{cases}
    \sum\limits_{i=1}^n x_i \cdot g_i = 0\\
    x \in \ZZ^n \setminus \{\BZero\}.
    \end{cases}\label{group_SVP_def}
\end{gather}

\begin{remark}\label{group_compl_rm}
Note that since $\abs{\det(S)} = \Delta$, the diagonal of $S$ contains at most $\log_2(\Delta)$ of elements that are not equal $1$. Therefore, the arithmetic complexity of one operation with elements of $\GC$ is $O(\log(\Delta))$.

\end{remark}

\begin{remark}\label{duplicates_SVP_rm}
W.l.o.g. we can assume that $g_i \not= \pm g_j$, for different $i,j$. Consequently, $n \leq \Delta/2 + 1$. Definitely, if for example $g_1 = \pm g_2$, then the vector $(1,\mp 1,0,\dots,0)^\top \in \ZZ^n\setminus\{\BZero\}$ is a feasible solution of \eqref{group_SVP_def}. Clearly, the only solutions, which can be better, are solutions of the type $\pm e_i$, which are be feasible only if $g_i = 0$.

The duplicates and zeros inside $g_1, g_2, \dots, g_n$ can be detected by an algorithm, like the radix-sort using $O(n \cdot \Delta)$ group operations or by any comparison-based sorting using $O(n \cdot \log(n))$ group operations.
\end{remark}

To solve the problem \eqref{group_SVP_def}, we will use the following dynamic programming scheme. For $g_0 \in \GC$ and $k \in \intint{n}$, we define the problem $\DPR(k,g_0)$ in the following way:
\begin{gather}
    \sum\limits_{i = 1}^k f_i(x_i) \to \min\notag\\
    \begin{cases}
    \sum\limits_{i=1}^k x_i \cdot g_i = g_0\\
    x \in \ZZ^k \setminus \{\BZero\}.
    \end{cases}\label{SVP_group_DP_def}
\end{gather}
Clearly, the problem $\DPR(n,0)$ is equivalent to the problem \ref{group_SVP_def}.

Denote 
\begin{gather*}
    \psi_+(k,g_0) = \min\limits_{j \in \intint[0]{\Delta}} \bigl\{ \DPR(k-1, g_0 - j \cdot g_k) + f_k(j) \bigr\},\\
    \psi_-(k,g_0) = \min\limits_{j \in \intint[0]{\Delta}} \bigl\{ \DPR(k-1, g_0 + j \cdot g_k) + f_k(j) \bigr\},\\
    \eta(k,g_0) = \min\bigl\{ f_k(j)\colon j \cdot g_k = g_0,\; j \in \intint[-\Delta]{\Delta}\setminus\{0\} \bigr\}.
\end{gather*}
If the set, where we take $\min$ for $\eta(k,g_0)$, is empty, then we set $\eta(k,g_0) = +\infty$.

Since $f_k$ is monotone and even, we have $\DPR(1,g_0) = \eta(1,g_0)$. Similarly, for $k \geq 2$, it can be straightforwardly checked out that
\begin{equation}
    \DPR(k,g_0) = \min\bigl\{\psi_+(k-1,g_0),\;\psi_-(k-1,g_0),\;\eta(k-1,g_0)\bigr\}.
\end{equation}
Note that we can not only use the values of $\psi_+(k-1,g_0)$ and $\psi_-(k-1,g_0)$ in the previous formula, because in this case we are missing out the solutions of the type $(0,0,\dots,0,j)^\top \in \ZZ^k$. So, we need additionally to take into account the values of $\eta(k-1,g_0)$.

First of all, fix $k$. Let us show how to compute the values $\eta(k,g_0)$, for all $g_0 \in \GC$, using $O(\Delta)$ group operations. Note that $\eta(k,g_0) \not= +\infty$ if and only if $g_0 = g_k \cdot j$, for some $j \in \ZZ\setminus\{0\}$. Hence, we need to fill only the values $\eta(k, j \cdot g_k)$, for other values of $g_0$ we can just set $\eta(k, g_0) = +\infty$. 

To fill $\eta(k, j \cdot g_k)$, we can do the following:
\begin{enumerate}
    \item Assign $\eta(k,g_0) := +\infty$, for all $g_0 \in \GC$;
    \item For $j \in \intint{\Delta-1}$, do the following:
    \item \qquad If $\eta(k,j \cdot g_k) = +\infty$, then assign $\eta(k,j \cdot g_k) := f_k(j)$;
    \item \qquad If $\eta(k,-j \cdot g_k) = +\infty$, then assign $\eta(k,-j \cdot g_k) := f_k(j)$;
\end{enumerate}
To see that the algorithm is correct, assume that $j^* \in \ZZ\setminus\{0\}$ is the value such that $g_0 = j^* \cdot g_k$ and $\abs{j^*}$ is minimal. Then, clearly $\eta(k, g_0) = f_k(j^*)$. If $j^* > 0$, we will find it during the $3$-th step. If $j^* < 0$, the $4$-th step will give the correct value. This value will not be rewritten, because only the values with $\eta(k,g_0) = +\infty$ could be assigned to something.

Therefore, the values $\eta(k,g_0)$, for $k \in \intint n$ and $g_0 \in \GC$, can be computed using $O(n \cdot \Delta)$ group operations.

Fix $k \geq 2$. Let us estimate the complexity to compute $\psi_+(k,\cdot)$ in assumption that the layer $\DPR(k-1,\cdot)$ is already computed.

Let us consider the quotient group $\QS = \GC/\langle g_k \rangle$ and fix $\QC \in \QS$. Let $d_k = \abs{\langle g_k \rangle}$. Clearly, $\QC = q + \langle g_k \rangle$, where $q \in \GC$ is a representative of $\QC$, and $d_k = \abs{\QC}$. Let us define the sequences $\{a_i\}_{i \in \intint[0]{d_k-1}}$, $\{b_i\}_{i \in \intint[0]{d_k-1}}$, and $\{c_i\}_{i \in \intint[0]{d_k-1}}$ as follows:
\begin{gather*}
    a_i = \DPR(k-1, q + i \cdot g_k),\\
    b_i = f_k(i),\\
    c_i = \psi_+(k, q + i \cdot g_k).
\end{gather*}

Assuming that $a_i = b_i = c_i = 0$ for $i < 0$, and due to the definition of $\psi_+$, we have
\begin{equation*}
    c_i = \min\limits_{j \in \intint[0]{i}} \bigl\{ a_{i-j} + b_i \bigr\}.
\end{equation*}
That gives $c = (a \star b)[0,d_k-1]$. Therefore, considering all the cosets $\QC \in \QS$, the group operations complexity to compute $\psi_+(k,\cdot)$, in the assumption that the level $\DPR(k-1,\cdot)$ has already been computed, can be expressed by
$$
O\bigl(d_k \cdot T_{conv}(\Delta/d_k)\bigr),
$$ where $T_{conv}(\cdot)$ denotes the complexity of the $(\min,+)$-convolution. The values of $\psi_-(k,\cdot)$ can be computed in a similar way with the same complexity bound.

Consequently, the layer $\DPR(k,\cdot)$, again in the assumption that the level $\DPR(k-1,\cdot)$ has already known, can be computed, using
$
O\bigl(d_k \cdot T_{conv}(\Delta/d_k)\bigr)
$ group operations.

The total group operations complexity is 
$$
O\bigl( n \cdot \Delta + \sum\limits_{k=1}^n d_k \cdot T_{conv}(\Delta/d_k) \bigr).
$$

Since $f_k$ is convex, due to \cite{GraphKnap}, $T_{conv}(k) = O(k)$. Consequently, the last bound becomes
$$
O(n \cdot \Delta).
$$

Due to Remark \ref{group_compl_rm}, the arithmetic complexity of group operations is $O(\log(\Delta))$. Hence, the total arithmetic complexity to solve the problem \eqref{group_SVP_def} can be expressed by
$$
O\bigl(n \cdot \Delta \cdot \log(\Delta) \bigr).
$$

Finally, assuming that the original group problem \eqref{group_SVP_def} contains duplicates, we can remove them, using Remark \ref{duplicates_SVP_rm} with $O\bigl(n \cdot \log(n) \cdot \log(\Delta)\bigr)$ arithmetic operations. Denoting the dimension of the resulting problem by $m \leq \Delta/2 + 1$ and taking into account the SNF computational complexity, denoted by $T_{SNF}(n,\Delta)$, we get the complexity bound of the whole algorithm
\begin{equation*}
    O\bigl( T_{SNF}(n,\Delta) + n \cdot \log(n) \cdot \log(\Delta) + m \cdot \Delta \cdot \log(\Delta) \bigr).
\end{equation*}

Due to \cite{SNFOptAlg}, $T_{SNF}(n,\Delta) = O\bigl(n^{\omega} \cdot \log(\Delta)\bigr)$ of arithmetic operations with integers of the size $O(\log(\Delta))$. So, the following theorem has been proven.

\begin{theorem}\label{main_SVP_th}
The problem \ref{gen_SVP_def} can be solved by an algorithm with arithmetic complexity bound
\begin{equation*}
    O\bigl( n^{\omega} \cdot \log(\Delta) + \min\{n,\Delta\} \cdot \Delta \cdot \log(\Delta)\bigr) = \tilde O(n^{\omega} + \min\{n,\Delta\} \cdot \Delta),\\
\end{equation*}
where all the computations are performed with integer numbers of the size $O(\log(\Delta))$ and $\omega$ is the matrix multiplication exponent.
\end{theorem}

\subsection{\ref{CVP_def} problem}\label{CVP_subs}

Let us consider the generalized problem \ref{gen_CVP_def}.

After the maps $x \to x - b$ and $q-b \to q$, where $b = -\lfloor q \rceil$, the problem \ref{gen_CVP_def} transforms to:
\begin{equation}\label{gen_CVP_redef}
    \min\Bigl\{ \sum\limits_{i=1}^n f\bigl( \abs{x_i - q_i} \bigr) \colon x \in b + \inth(A) \Bigr\},
\end{equation}
with $\norm{q}_{\infty} < 1/2$. Additionally, we can assume that $q \geq 0$, because we can map $x_i$ to $-x_i$, for $q_i < 0$. Finally, we can sort $q_i$, so we have
\begin{equation}\label{qi_mono_eq}
1/2 > q_1 \geq q_2 \geq \dots \geq q_n \geq 0.
\end{equation}

Denote $f_i(x) = f\bigl(\abs{x - q_i}\bigr)$. It is easy to check that the following properties hold for $f_i(x)$:
\begin{enumerate}
    \item For any $i \in \intint n$, $f_i$ is monotone on the sets $\ZZ_{\geq 0}$ and $\ZZ_{\leq 0}$ and convex on $\ZZ$;
    \item For any $x \in \ZZ_{\geq 1}$, 
\begin{equation*}
    f_1(x) \leq f_2(x) \leq \dots \leq f_n(x);
\end{equation*}
    \item For any $x \in \ZZ_{ \leq 0}$,
    \begin{equation*}
    f_1(x) \geq f_2(x) \geq \dots \geq f_n(x).
\end{equation*}
    \item For any $x \in \ZZ_{\geq 1}$,
    \begin{equation*}
        f_1(x) - f_1(x-1) \leq f_2(x) - f_2(x-1) \leq \dots \leq f_n(x) - f_n(x-1);
    \end{equation*}
    \item For any $x \in \ZZ_{\leq 0}$,
    \begin{equation*}
        f_1(x-1) - f_1(x) \geq f_2(x-1) - f_2(x) \geq \dots \geq f_n(x-1) - f_n(x).
    \end{equation*}
\end{enumerate}

The property 1 could be checked directly. The properties 2,3 hold, due to \eqref{qi_mono_eq} and the the monotonicity of $f$. The properties 4,5 hold, due to \eqref{qi_mono_eq} and the convexity of $f$.  

As in Subsection \ref{SVP_subs}, using the SNF decomposition $P A Q = S$, we transform \eqref{gen_CVP_redef} to:
\begin{gather}
    \sum\limits_{i=1}^n f_i(x_i) \to \min\notag\\
    \begin{cases}
    P x \equiv P b \pmod{S \cdot \ZZ^n}\\
    x \in \ZZ^n.
    \end{cases}\label{SNF_CVP_def}
\end{gather}

Let us define $\GC$ and $g_i$ (for $i \in \intint n$) as it was done in Subsection \ref{SVP_subs}. Let us define $G = P b \bmod \diag(S)$ and reformulate \eqref{SNF_CVP_def} in the group minimization style:
\begin{gather}
    \sum\limits_{i=1}^n f_i(x_i) \to \min\notag\\
    \begin{cases}
    \sum\limits_{i=1}^n g_i \cdot x_i = G\\
    x \in \ZZ^n.
    \end{cases}\label{group_CVP_def}
\end{gather}

Now, we going to remove duplicates from $g_1, g_2, \dots, g_n$, but it is a bit more tricky problem than its analogue discussed in Remark \ref{duplicates_SVP_rm}. Assume that $g_1 = g_2 = \dots = g_k$.  We want to replace the variables $x_1, \dots, x_k$ by only one variable $y = x_1 + \dots + x_k$ attached to $g_1$. To this end, we need to replace the objective $\sum_{i=1}^k f_i(x_i)$ with a new equivalent objective $h(y)$. The following lemmas explain how to choose $h$.

\begin{lemma}\label{duplicates_elim_lm}
Let $g_1 = g_2 = \dots = g_k$, for $k \in \intint n$, and $x^*$ be an optimal solution of \ref{group_CVP_def}. Denote $S = x_1^* + \dots + x_k^*$.

Then, there exists an optimal solution $z^*$ with the following structure:
\begin{enumerate}
    \item If $S \geq 0$, then:
    \begin{equation}\label{greed_look_positive}
    z^* = (a+1, a+1, \dots, a+1, a , a, \dots, a)^\top,
    \end{equation}
    where $a \in \ZZ_{\geq 0}$.
    
    \item If $S < 0$, then:
    \begin{equation}\label{greed_look_negative}
    z^* = -(a, a, \dots, a, a+1 , a+1, \dots, a+1)^\top,    
    \end{equation}
    where $a \in \ZZ_{\geq 0}$.
\end{enumerate}
\end{lemma}
\begin{proof}
Note that the expressions $g_1 \cdot x_1^* + \dots + g_k \cdot x_k^*$ and $g_1 \cdot S$ are equivalent in therms of constraints of $\eqref{group_CVP_def}$.

Assume that $S \geq 0$. First of all, we claim that there exists an optimal solution $z^*$ with the property $z^*_i \geq 0$, for $i \in \intint k$. Assume that there exist $i ,j \in \intint k$ with $x^*_i \geq 1$ and $x^*_j \leq -1$. Since $S \geq 0$, if $x^*_j$ exists, then $x^*_i$ exists also. Next, we construct a vector $z^*$, which coincides with $x^*$ in all the coordinates, except $i,j$. Put $z^*_i = x^*_i - 1$ and $z^*_j = x^*_j + 1$. Due to Property 1, we have $\sum_{i=1}^k f_i(z^*_i) \leq \sum_{i=1}^k f_i(x^*_i)$. Such a procedure can be repeated until no negative coordinates remain. Consequently, it can be assumed that $x^*_i \geq 0$, for $i \in \intint k$.

Let us consider the following auxiliary optimization problem:
\begin{gather*}
    \sum_{i=1}^k f_i(x_i) \to \min\notag\\
    \begin{cases}
    x_1 + \dots + x_k = S\\
    x \in \ZZ_{\geq 0}.
    \end{cases}\label{aux_CVP_prob}
\end{gather*}
Clearly, $x^*[1,k]$ gives an optimal solution of this problem, and vice versa, an optimal solution of \eqref{aux_CVP_prob} could be used to generate the first $k$ coordinates of $x^*$. 

Let us consider the set 
$$
\SC = \{x \in \ZZ^k_{\geq 0} \colon x_1 + \dots +x_k \leq S\}. 
$$ Elements of $\SC$ could be treated as the characteristic vectors of multisets with the cardinality $S$. Identifying vectors with multisets, we can see that $\SC$ is a matroid, see, for example, \cite[Proposition~13.4, Part~13.~Matroids]{KorteBook}. The vectors $z \in \SC$ with $z_1 + \dots + z_k = S$ are the bases of $\SC$. Consequently, an optimal solution of \eqref{aux_CVP_prob} is exactly a base of $\SC$ with the minimal possible value of the objective function.

Since $\SC$ is a matroid, an optimal solution of \eqref{aux_CVP_prob} can by found by the following greedy algorithm:
\begin{enumerate}
    \item Assign $s := 0$, $x := \BZero^k$, and $F : = f_1(0) + \dots + f_k(0)$;
    \item While $s \leq S$ do the following:
    \item \qquad Choose $i \in \intint k$, such that the value $f_i(x_i+1) - f_i(x_i)$ is minimal;
    \item \qquad Assign $x_i := x_i + 1$, $s := s+1$, and $F := F + f_i(x_i+1) - f_i(x_i)$;
    \item \qquad Move to the step 2;
    \item Return $x$ as a greedy solution and $F$ as $f(x)$;
\end{enumerate}

Due to the properties 4,5, there exists a greedy solution $z^*$ that looks like \eqref{greed_look_positive}. This proves the lemma for the case $S \geq 0$. The case $S < 0$ is absolutely similar. $\Box$


\end{proof}

\begin{lemma}
Let $g_1 = g_2 = \dots = g_k$, for $k \in \intint n$. There exists a function $h(x) \colon \ZZ \to \RR$, such that the problem \eqref{group_CVP_def} and the following problem
\begin{gather}
    h(y) + \sum\limits_{i = k+1}^n f_i(x_{i-k}) \to \min\notag\\
    \begin{cases}
    g_1 \cdot y + \sum\limits_{i = k+1}^n g_i \cdot x_{i-k} = G\\
    x \in \ZZ_{\geq 0}^{n-k}\\
    y \in \ZZ_{\geq 0}
    \end{cases}\label{reduced_size_CVP_def}
\end{gather} are equivalent.

The function $h$ can be defined in the following way:
\begin{enumerate}
    \item If $y \in \ZZ_{\geq 0}$, then compute $r = y \bmod k$ and $a = \lfloor y/k \rfloor$. Let us construct the vector
    $$
    z = (a+1, \dots,a+1, a, \dots, a)^\top,
    $$ where $a+1$ is taken $r$ times. Put $h(y) = \sum\limits_{i=1}^k f_i(z_i)$;
    
    \item If $y \in \ZZ_{< 0}$, then compute $r = (-y) \bmod k$ and $a = \lfloor (-y)/k \rfloor$. Let us construct the vector
    $$
    z = -(a, \dots,a, a+1, \dots, a+1)^\top,
    $$ where $a+1$ is taken $r$ times. Put $h(y) = \sum\limits_{i=1}^k f_i(z_i)$.
\end{enumerate}

Additionally, for any $m \geq 0$, the sequences 
\begin{gather*}
    h(0),\,h(1),\, \dots,\, h(m) \quad \text{and}\\
    h(0),\,h(-1),\, \dots,\, h(-m)
\end{gather*} can be computed using $O(m)$ arithmetic operations.
\end{lemma}
\begin{proof}
If $x^*$ is an optimal solution of \eqref{group_CVP_def}, then, due to Lemma \ref{duplicates_elim_lm}, there exists an optimal solution $z^*$ of \eqref{group_CVP_def}, such that the first $k$ components of $z^*$ look like \eqref{greed_look_positive} or \eqref{greed_look_negative}. By the definition of $h$, we have $\sum_{i=1}^k f_i(z^*_i) = h(y)$, where $y = z_1^* + \dots z_k^*$. Note that $(y,\, x^*[k+1,n])$ is a feasible solution of \eqref{reduced_size_CVP_def} with the same value of the objective function.

In the opposite direction, let $(y,\, x^*)$ be an optimal solution of \eqref{reduced_size_CVP_def}. Let us construct the vector $z$ as it was described in the lemma definition. Clearly, the vector $\binom{z}{x^*}$ is a feasible solution of \eqref{group_CVP_def} with the same value of the objective function.

Finally, let us explain how to compute $h(0),h(1), \dots, h(m)$ with $O(m)$ operations. Assume that $h(i)$  has already been computed, and let $r = i \bmod k$ and $a = \lfloor i / k \rfloor$. Then, by the definition of $h$, we have $h(i + 1) = h(i) + f_{r+1}(a+1) - f_{r+1}(a)$. Hence, we need $O(1)$ operations to compute $h(i+1)$ and $O(m)$ operations to compute the whole sequence. A similar algorithm works for $h(0), h(-1), \dots, h(-m)$. $\Box$
\end{proof}

Using the previous lemma, we can remove all the duplicates and assume that all the elements $g_1, g_2, \dots, g_n$ are unique.

The remaining part of our algorithm is very close to the algorithm from Subsection \ref{SVP_subs}. For $k \in \intint n$ and $g_0 \in \GC$, we define $\DPR(k,g_0)$, $\psi_+(k,g_0)$ and $\psi_-(k,g_0)$. Clearly, the problem $\DPR(n,G)$ is equivalent to the problem \ref{group_CVP_def}. The values $\psi_+(k,g_0)$ and $\psi_-(k,g_0)$ can be computed with the same formulas and algorithms. The only minor difference is the recurrent formula for $\DPR(k,g_0)$:
$$
\DPR(k,g_0) = \min\bigl\{ \psi_+(k-1,g_0), \psi_-(k-1,g_0) \bigr\}.
$$

Therefore, we have proven our conclusive result.
\begin{theorem}\label{main_CVP_th}
The problem \ref{gen_CVP_def} can be solved by an algorithm with arithmetic complexity bound
\begin{equation*}
    O\bigl( n^{\omega} \cdot \log(\Delta) + \min\{n,\Delta\} \cdot \Delta \cdot \log(\Delta)\bigr) = \tilde O(n^{\omega} + \min\{n,\Delta\} \cdot \Delta),\\
\end{equation*}
where all computations are performed with integer numbers of the size $O(\log(\Delta))$ and $\omega$ is the matrix multiplication exponent.
\end{theorem}

\section*{Acknowledgments}

The article was prepared within the framework of the Basic Research Program at the National Research University Higher School of Economics (HSE). 

The authors would like to thank A.~Vanin for useful discussions during preparation of this article.

\bibliographystyle{spmpsci}
\bibliography{grib_biblio}

\end{document}